\documentclass[a4paper,UKenglish,cleveref, autoref, thm-restate]{lipics-v2019}

\title{An Investigation of the Recoverable Robust Assignment Problem}

\author{Dennis Fischer}{RWTH Aachen, Department of Computer Science, Germany}{fischer@algo.rwth-aachen.de}{}{}
\author{Tim A. Hartmann}{RWTH Aachen, Department of Computer Science, Germany}{hartmann@algo.rwth-aachen.de}{}{}
\author{Stefan Lendl}{Graz University of Technology, Institute of Discrete Mathematics, Austria}{lendl@math.tugraz.at}{}{}
\author{Gerhard J. Woeginger}{RWTH Aachen, Department of Computer Science, Germany}{woeginger@algo.rwth-aachen.de}{}{}
\authorrunning{D. Fischer, T. A. Hartmann, S. Lendl, G. J. Woeginger } 

\Copyright{John Q. Public and Joan R. Public} 



\category{} 

\relatedversion{} 

\supplement{}

\funding{Stefan Lendl acknowledges the support of the Austrian Science Fund (FWF): W1230. Dennis
Fischer and Gerhard J. Woeginger have been supported by the DFG RTG 2236 ``UnRAVeL''.}

\acknowledgements{
We thank Bettina Klinz for several helpful discussions about the topic.
}

\nolinenumbers 

\EventEditors{John Q. Open and Joan R. Access}
\EventNoEds{2}
\EventLongTitle{42nd Conference on Very Important Topics (CVIT 2016)}
\EventShortTitle{CVIT 2016}
\EventAcronym{CVIT}
\EventYear{2016}
\EventDate{December 24--27, 2016}
\EventLocation{Little Whinging, United Kingdom}
\EventLogo{}
\SeriesVolume{42}
\ArticleNo{23}

\usepackage{complexity}
\usepackage{amsmath}
\usepackage{mathtools}
\usepackage{float}

\usepackage{xcolor}
\definecolor{RWTH_blue100}{HTML}{00549F}
\definecolor{RWTH_blue75}{HTML} {407FB7}
\definecolor{RWTH_blue50}{HTML} {8EBAE5}
\definecolor{RWTH_blue25}{HTML} {C7DDF2}
\definecolor{RWTH_blue10}{HTML} {E8F1FA}
\definecolor{RWTH_black100}{HTML}{000000}
\definecolor{RWTH_black75}{HTML} {646567}
\definecolor{RWTH_black50}{HTML} {9C9E9F}
\definecolor{RWTH_black25}{HTML} {CFD1D2}
\definecolor{RWTH_black10}{HTML} {ECEDED}
\definecolor{RWTH_magenta100}{HTML}{E30066}
\definecolor{RWTH_magenta75}{HTML} {E96088}
\definecolor{RWTH_magenta50}{HTML} {F19EB1}
\definecolor{RWTH_magenta25}{HTML} {F9D2DA}
\definecolor{RWTH_magenta10}{HTML} {FDEEF0}
\definecolor{RWTH_yellow100}{HTML}{FFED00}
\definecolor{RWTH_yellow75}{HTML} {FFF055}
\definecolor{RWTH_yellow50}{HTML} {FFF59B}
\definecolor{RWTH_yellow25}{HTML} {FFFAD1}
\definecolor{RWTH_yellow10}{HTML} {FFFDEE}
\definecolor{RWTH_petrol100}{HTML}{006165}
\definecolor{RWTH_petrol75}{HTML} {2D7F83}
\definecolor{RWTH_petrol50}{HTML} {7DA4A7}
\definecolor{RWTH_petrol25}{HTML} {BFD0D1}
\definecolor{RWTH_petrol10}{HTML} {E6ECEC}
\definecolor{RWTH_tuerkis100}{HTML}{0098A1}
\definecolor{RWTH_tuerkis75}{HTML} {0098A1}
\definecolor{RWTH_tuerkis50}{HTML} {89CCCF}
\definecolor{RWTH_tuerkis25}{HTML} {CAE7E7}
\definecolor{RWTH_tuerkis10}{HTML} {EBF6F6}
\definecolor{RWTH_green100}{HTML}{57AB27}
\definecolor{RWTH_green75}{HTML} {8DC060}
\definecolor{RWTH_green50}{HTML} {B8D698}
\definecolor{RWTH_green25}{HTML} {DDEBCE}
\definecolor{RWTH_greem10}{HTML} {F2F7EC}
\definecolor{RWTH_orange100}{HTML}{F6A800}
\definecolor{RWTH_orange75}{HTML} {FABE50}
\definecolor{RWTH_orange50}{HTML} {FDD48F}
\definecolor{RWTH_orange25}{HTML} {FEEAC9}
\definecolor{RWTH_orange10}{HTML} {FFF7EA}
\definecolor{RWTH_red100}{HTML}{CC071E}
\definecolor{RWTH_red75}{HTML} {D85C41}
\definecolor{RWTH_red50}{HTML} {E69679}
\definecolor{RWTH_red25}{HTML} {F3CDBB}
\definecolor{RWTH_red10}{HTML} {FAEBE3}
\definecolor{RWTH_violett100}{HTML}{612158}
\definecolor{RWTH_violett75}{HTML} {834E75}
\definecolor{RWTH_violett50}{HTML} {A8859E}
\definecolor{RWTH_violett25}{HTML} {D2C0CD}
\definecolor{RWTH_violett10}{HTML} {EDE5EA}
\usepackage{tikz}
\usetikzlibrary{arrows,intersections}
\usetikzlibrary{shapes,snakes}
\usetikzlibrary{calc}
\usetikzlibrary{shapes}
\usetikzlibrary{shapes.misc}
\usetikzlibrary{decorations.pathreplacing}
\usetikzlibrary{intersections,positioning,calc}
\usetikzlibrary{patterns}
\usetikzlibrary{hobby}
\usetikzlibrary{decorations}

\pgfdeclarelayer{bg}
\pgfsetlayers{bg,main}

\tikzset{
	thick/.style = {line width=.75mm},
	matched/.style = {red!90,thick}, 
	matchedalt/.style = {RWTH_green75!80, thick},
}

\newcommand{\nn}{\mathbb{N}}

\newcommand{\rr}{\mathbb{R}}

\newcommand{\probl}[1]{\textsc{#1}}

\newcommand{\set}[1]{\left\{#1\right\}}

\newcommand{\doublematchedH}[3]
{
\begin{pgfonlayer}{bg}
		\draw[transform canvas={yshift=.11em}, matched] (#2.center) -- (#3.center);
		\draw[transform canvas={yshift=-.11em}, matchedalt] (#2.center) -- (#3.center);
\end{pgfonlayer}
}

\newcommand{\doublematchedV}[3]
{
\begin{pgfonlayer}{bg}
		\draw[transform canvas={xshift=.11em}, matched] (#2.center) -- (#3.center);
		\draw[transform canvas={xshift=-.11em}, matchedalt] (#2.center) -- (#3.center);
\end{pgfonlayer}
}

\newcommand{\problemdefSimple}[3]{
\vspace{0.15cm}
\begin{tabularx}{{0.95\textwidth}}{ r X }
\multicolumn{2}{l}{#1} \\
Input: & #2 \\
Question: & #3
\end{tabularx}
\vspace{0.1cm}
}

\begin{document}
\maketitle

\begin{abstract}
	We investigate the so-called recoverable robust assignment problem on balanced
	bipartite graphs with $2n$ vertices, a mainstream problem in robust optimization:
	For two given linear cost functions $c_1$ and $c_2$ on the edges and a given
	integer $k$, the goal is to find two perfect matchings $M_1$ and $M_2$ that
	minimize the objective value $c_1(M_1)+c_2(M_2)$, subject to the constraint that 
	$M_1$ and $M_2$ have at least $k$ edges in common. 
	
	We derive a variety of results on this problem.
	First, we show that the problem is W[1]-hard with respect to the parameter $k$,
	and also with respect to the recoverability parameter $k'=n-k$. 
	This hardness result holds even in the highly restricted special case where
	both cost functions $c_1$ and $c_2$ only take the values $0$ and $1$.
	(On the other hand, containment of the problem in XP is straightforward to see.)
	Next, as a positive result we construct a polynomial time algorithm for the 
	special case where one cost function is Monge, whereas the other one is Anti-Monge. 
	Finally, we study the variant where matching $M_1$ is frozen, and where the
	optimization goal is to compute the best corresponding matching $M_2$,
	the second stage recoverable assignment problem.
	We show that this problem variant is contained in the randomized parallel
	complexity class $\text{RNC}_2$, and that it is at least as hard as the infamous 
	problem \probl{Exact Matching in Red-Blue Bipartite Graphs} whose computational 
	complexity is a long-standing open problem.
	\end{abstract}
	
\newpage
	
	\section{Introduction}
	\label{sec:intro}
	The \probl{Assignment Problem} (\probl{AP}) is a fundamental and well-investigated problem 
	in discrete optimization:
	For the complete bipartite graph $K_{n,n}=(V,E_{n,n})$ with given costs $c:E_{n,n}\to\rr$ on the edges, 
	the \probl{AP} asks for a perfect matching $M$ in $K_{n,n}$ that minimizes the total cost $c(M)$.
	The \probl{AP} can be solved in polynomial time, by using for instance the Hungarian method or techniques 
	from network flow theory; see Burkard, Dell'Amico \& Martello \cite{burkard2012assignment}.
	
	In this paper we study a variant of the \probl{AP} from the area of robust optimization,
	which we denote as \probl{recoverable assignment problem} (\probl{RecovAP}).
	An instance of \probl{RecovAP} consists of two cost functions $c_1,c_2:E_{n,n}\to\rr$ on the edges of $K_{n,n}$ 
	together with an integer bound $k$.
	The goal is to find two perfect matchings $M_1$ and $M_2$ that minimize the objective value $c_1(M_1)+c_2(M_2)$, 
	subject to the constraint that $M_1$ and $M_2$ have at least $k$ edges in common.
	We also consider the following two non-trivial special cases of \probl{RecovAP}:
	\begin{itemize}
	\item
	Consider an arbitrary (bipartite) subgraph $G=(V,E)$ of $K_{n,n}$. 
	If the cost functions $c_1$ and $c_2$ are set to $+\infty$ on all edges outside $E$, 
	one arrives at the graphic special case of \probl{RecovAP} for bipartite input graphs $G$.
	This allows us to study the problem with graph-theoretic tools, and to look into graph-theoretic structures.
	\item
	If the cost function $c_1$ is set to zero on the edges of some fixed perfect matching and set to $+\infty$
	on all the remaining edges, the perfect matching $M_1$ is thereby fixed and frozen at the zero-cost edges.
	Then problem \probl{RecovAP} boils down to finding a matching $M_2$ that minimizes $c_2(M_2)$ subject
	to the constraint $|M_1\cap M_2|\ge k$; the resulting optimization problem is called the
	\probl{Second-stage recoverable assignment problem} (\probl{2S-RecovAP}).
	\end{itemize}
	Both problems \probl{RecovAP} and \probl{2S-RecovAP} are motivated by (central and natural) questions in the
	area of Recoverable Robust Optimization; see Appendix~\ref{sec:robustapplication} for more details.
	
	\subparagraph{Known and related results.}
	The study of discrete optimization problems with intersection constraints (as imposed in problem \probl{RecovAP}) 
	was initiated through applications in Recoverable Robust Optimization under interval uncertainty.
	The literature mainly analyzes situations where the feasible solutions form the bases of varous types of matroids: 
	Kasperski \& Zieli\'{n}ski~\cite{kasperski2015robust} construct a polynomial time solution for the case 
	of uniform matroids; the underlying robust optimization problem is called the recoverable selection problem.
	Lachmann \& Lendl~\cite{lachmann2019} provide a simple greedy-type algorithm for recoverable selection, and 
	thereby improve the time complexity in \cite{kasperski2015robust} from cubic time down to linear time.
	Hradovic, Kasperski \& Zieli\'{n}ski~\cite{hradovich2017recoverable,hradovich2017recoverable-MST}
	obtain a polynomial time algorithm for the recoverable matroid basis problem and a 
	strongly polynomial time algorithm for the recoverable spanning tree problem. 
	These results have been generalized and improved by Lendl, Peis \& Timmermans~\cite{lendl2019matroid} 
	who show that the recoverable matroid basis and the recoverable polymatroid basis problem can both be 
	solved in strongly polynomial time. 
	Iwamasa \& Takayawa~\cite{iwamasa2020conf} further generalize these results and cover cases with nonlinear 
	and convex cost functions.
	
	B\"using \cite{busing2012recoverable} derives various NP-hardness results for recoverable 
	robust shortest $s$-$t$-path problems, and thus makes one of the first steps in this area 
	beyond feasible solutions with a matroidal structure.
	Note that the combinatorics of $s$-$t$-paths is substantially more complex than the combinatorics of matroid bases: 
	whereas all bases of a matroid have the same cardinality, different $s$-$t$-paths may contain totally
	different numbers of edges.
	For that reason (see also Appendix~\ref{sec:robustapplication}), recoverable robust shortest $s$-$t$-path 
	problems do not (easily) translate into corresponding optimization problems that ask for two feasible solutions 
	with at least $k$ common elements.

	Şeref et al.~\cite{seref-incremnetal-ap} study \probl{2S-RecovAP} and obtain a randomized algorithm 
	running in polynomial time if the costs are polynomially bounded.
	
	\subparagraph{Our contribution.}
	By analyzing problem \probl{RecovAP}, we take another step beyond matroidal structures 
	in recoverable robust optimization.
	Section~\ref{sec:parameterized} discusses the computational complexity of \probl{RecovAP}.
	We look into the parameterized complexity of \probl{RecovAP}.
	We show that the problem is W[1]-hard with respect to the central parameter $k$, the lower bound on the 
	intersection size of the two matchings.
	Furthermore, the problem is W[1]-hard with respect to the so-called recoverability parameter $k'=n-k$, 
	which bounds the number of edges that are in matching $M_1$ but not in matching $M_2$.
	These hardness results even hold in the highly restricted case where both cost functions $c_1$ and $c_2$ 
	only take the values $0$ and $1$.
	Similar W[1]-hardness results hold for the graphic version of \probl{RecovAP} on planar graphs.
	On the positive side, there exists a simple XP algorithm for parameter $k$ (that checks all possible 
	sets $M_1\cap M_2$ of size $k$) and there also exists a simple XP algorithm for parameter $k'$ (that 
	checks all possible sets $M_1-M_2$ and $M_2-M_1$ of size $k'$). This is in contrast to the 
	variants of the problem with the constraint $|M_1 \cap M_2| \leq k$ or  $|M_1 \cap M_2| = k$.
	These problems are easily shown to be NP-hard for each fixed $k$ via a reduction from 
	the \probl{Disjoint Matchings Problem}~\cite{Frieze1983}.
	Finally, the graphic version of \probl{RecovAP} with respect to parameter treewidth is in FPT.
	
	Next, in Section~\ref{sec:monge} we discuss problem \probl{RecovAP} under Monge-type conditions; we refer 
	to Burkard, Klinz \& Rudolf~\cite{burkard1996perspectives} for an extensive overview of Monge properties.
	The cost function in the assignment problem may naturally be viewed as an $n\times n$ cost matrix.
	If both cost functions $c_1$ and $c_2$ correspond to Monge matrices, problem \probl{RecovAP} boils 
	down to something trivial:
		In the optimal solution both matchings $M_1$ and $M_2$ run along the 
	main diagonal of the underlying matrix.
	And if both cost functions $c_1$ and $c_2$ correspond to Anti-Monge matrices, then in the optimal 
	solution both matchings $M_1$ and $M_2$ run along the secondary diagonal of the underlying matrix.
	The mixed case where $c_1$ corresponds to a Monge matrix and where $c_2$ corresponds to an Anti-Monge matrix
	is less trivial and more interesting.
	By analyzing the combinatorial structure of potential optimal solutions, we show that it is
	solvable in polynomial time.
	
	Finally, in Section~\ref{sec:incremental} we turn to the second-stage recoverable assignment 
	problem \probl{2S-RecovAP}, which shows a strange and rather unpleasant behavior.
	We feel that problem \probl{2S-RecovAP} is too hard to allow a polynomial time solution, and 
	we simultaneously feel that it is too easy to allow an NP-hardness proof.
	We support our intuition by two mathematical arguments:
	First, by a straightforward reduction to the extact matching problem in red blue biparite graphs 
	by Şeref et al.~\cite{seref-incremnetal-ap}, there exists an $\text{RNC}_2$ algorithm for \probl{2S-RecovAP}.
	As the complexity class $\text{RNC}_2\subseteq\text{RNC}$ is conjectured to be properly 
	contained in $\text{NP}$, this provides evidence for the easiness of \probl{2S-RecovAP}.
	Secondly, we show that the exact matching problem in red-blue bipartite graphs~\cite{exactmatchingweb} 
	is logspace reducible to \probl{2S-RecovAP}. 
	As the existence of a polynomial time algorithm for this exact red-blue matching problem is doubtful
	(and constitutes a long-open famous problem), this provides evidence for the hardness of \probl{2S-RecovAP}.

\section{Parameterized Complexity}\label{sec:parameterized}

To show $\W[1]$-hardness of the \probl{RecovAP} problem we reduce from the well known grid tiling problem. In the grid tiling problem we are given an $\ell \times \ell$ grid in which every cell contains a set of tuples.
The task is to select a value for every row and for every column compatible with the tuples in the cells:
	Each cell defined by a row and column combination contains a tuple with the values selected for this row and column.

\problemdefSimple{\probl{Grid Tiling}}
{Integers $\ell$, $n$, and a collection $\mathcal{S} = (S_{i, j})_{(i, j) \in [\ell]\times [\ell]}$ with $S_{i, j} \subseteq [n]\times [n]$.}
{Are there integers $r_{1},\dots,r_\ell$ and $c_{1},\dots,c_\ell$  such that $(r_i, c_j) \in S_{i, j}$ for every $i, j \in [\ell]$?}


\probl{Grid Tiling} has been shown to be $W[1]$-hard for parameter $\ell$ and has no   $f(\ell)n^{o(\ell)}$-time algorithm \cite{Cygan2015}.
For simplicity, we assume that every value $1, \ldots, n$ appears in at least one tuple of $\mathcal{S}$;
	which can be achieved by renaming the occurring $n$ many values increasingly.
That way we ensure that the size of the numbers are polynomial in the size of the input.

The main idea of the reduction is to build two graphs, one for the rows and one for the columns of the grid tiling problem. Those two graphs are identified at exactly those edges that represent the tuples. We have exactly one pair of edges in the intersection for every cell of the grid. We set the costs in such a way that an optimal solution has cost 0. That way certain edges are allowed or disallowed for a matching via the assigned costs.

\newcommand{\row}[4][n]
{
	\node[below of=#3Left] (#2Left) {};
	\node[right of=#2Left] (#2Right) {};
	\node[left of=#2Left] (#2Start) {};
	\node[right of=#2Right] (#2End) {};
	\begin{pgfonlayer}{bg}
		\ifthenelse{\equal{#1}{m}}
		{
			\draw (#2Start.center) --(#2Left.center);
			\draw[matched, thick] (#2Left.center) -- (#2Right.center);
			\draw (#2Right.center) -- (#2End.center);
		}
		{
			\draw[matched] (#2Start.center) --(#2Left.center);
			\draw[thick] (#2Left.center) -- (#2Right.center);
			\draw[matched] (#2Right.center) -- (#2End.center);
		}
	\end{pgfonlayer}
	
	\coordinate[label={[label distance=-0.25cm, align=center]above: \footnotesize #4}] (d) at ($(#2Left)!0.5!(#2Right)$);
}
\newcommand{\rowNone}[4][n]
{
	\node[below of=#3Left] (#2Left) {};
	\node[right of=#2Left] (#2Right) {};
	\node[left of=#2Left] (#2Start) {};
	\node[right of=#2Right] (#2End) {};
	\begin{pgfonlayer}{bg}
		\draw (#2Start.center) --(#2Left.center);
		\draw[thick] (#2Left.center) -- (#2Right.center);
		\draw (#2Right.center) -- (#2End.center);
	\end{pgfonlayer}
	
	\coordinate[label={[label distance=-0.25cm, align=center]above: \footnotesize #4}] (d) at ($(#2Left)!0.5!(#2Right)$);
}

\begin{figure}[t]
	\begin{center}

\captionsetup{position=b}
\begin{tikzpicture}[scale=1]
\begin{scope}[local bounding box=boxA, every node/.style={draw,circle,fill=white}, node distance=0.75cm]

\node (n1S) {};
\node[below of=n1S] (n2S) {};
\node[below of=n2S] (n3S) {};
\node[below of=n3S] (n4S) {};
\node[below of=n4S] (n5S) {};

\node[right of=n1S, label={right: $r_i = 1$}] (n1) {};
\node[right of=n2S, label={right: $\underline{r_i = 2}$}] (n2) {};
\node[right of=n3S, label={right: $r_i = 3$}] (n3) {};
\node[right of=n4S, label={right: $r_i = 4$}] (n4) {};
\node[right of=n5S, label={right: $r_i = 5$}] (n5) {};

\node[left of=n3S] (start) {};

\begin{pgfonlayer}{bg}    
\draw (start.center) --(n1S.center);
\draw (start.center) --(n2S.center);
\draw[matched] (start.center) --(n2S.center);
\draw (start.center) --(n3S.center);
\draw (start.center) --(n4S.center);
\draw (start.center) --(n5S.center);

\draw (n2S.center) -- (n2.center);
\draw (n1S.center) -- (n1.center);
\draw (n3S.center) -- (n3.center);
\draw (n4S.center) -- (n4.center);
\draw (n5S.center) -- (n5.center);

\draw[matched] (n1S.center) -- (n1.center);
\draw[matched] (n3S.center) -- (n3.center);
\draw[matched] (n4S.center) -- (n4.center);
\draw[matched] (n5S.center) -- (n5.center);

\node[right of=n1, above of=n1, draw=none] (gne) {};
\node[above of=n1S, draw=none] at ($(n1S)!0.5!(n1)$)  (gnw) {};
\node[below of=n5S, draw=none] at ($(n5S)!0.5!(n5)$)  (gsw) {};
\node[right of=n5, below of=n5, draw=none] (gse) {};

\draw [draw=black, dashed] (gnw.center) -- (gne.center);
\draw [draw=black, dashed] (gnw.center) -- (gsw.center);
\draw [draw=black, dashed] (gsw.center) -- (gse.center);

\end{pgfonlayer}

	\end{scope}
	\node[xshift=-.1cm, yshift=-1.55cm] at (boxA.west) {\textbf{(a)}};
	\begin{scope}[local bounding box=boxB, xshift=7cm, scale=0.8, every node/.style={draw,circle,fill=white, minimum size=.20cm}, node distance=0.75cm]

\coordinate (tLeft);
		
		\rowNone{13}{t}{$(1,3)$}
		\rowNone{14}{13}{$(1,4)$}
		\rowNone{17}{14}{$(1,7)$}
		\rowNone{21}{17}{$(2,1)$}
		\rowNone{22}{21}{$(2,2)$}
		
		\row{13}{t}{$(1,3)$}
		\row{14}{13}{$(1,4)$}
		\row{17}{14}{$(1,7)$}
		\row[m]{21}{17}{$(2,1)$}
		\row{22}{21}{$(2,2)$}
		
		\coordinate (tC) at ($(13Left)!0.5!(13Right)$);
		\node[draw=none, above=0.5cm of tC] {$S_{i,j}$};
		
		\node[left of=14Start] (Choice1) {};
		\node[left of=Choice1, label={above left: $r_i = 1$}] (Choice1L) {};

		\node[right of=14End] (Choice1E) {};
		\node[right of=Choice1E, label={}] (Choice1ER) {};

		\coordinate(2Middle) at ($(21Start)!0.5!(22Start)$);
		\node[left of=2Middle] (Choice2) {};
		\node[left of=Choice2, label={below left: $\underline{r_i = 2}$}] (Choice2L) {};

		\coordinate(2MiddleE) at ($(21End)!0.5!(22End)$);
		\node[right of=2MiddleE] (Choice2E) {};
		\node[right of=Choice2E, label={}] (Choice2ER) {};

		\coordinate(Choice1Middle) at ($(Choice1)!0.5!(Choice1L)$);
		\coordinate(Choice2MiddleE) at ($(Choice2E)!0.5!(Choice2ER)$);
		
		\begin{pgfonlayer}{bg}    
		\draw (Choice1L.center) --(Choice1.center);
		\draw[matched] (Choice1L.center) --(Choice1.center);
		
		\draw (Choice1.center) --(13Start.center);
		\draw (Choice1.center) --(14Start.center);
		\draw (Choice1.center) --(17Start.center);
		\draw (Choice1ER.center) --(Choice1E.center);
		\draw[matched] (Choice1ER.center) --(Choice1E.center);
		
		\draw (Choice1E.center) --(13End.center);
		\draw (Choice1E.center) --(14End.center);
		\draw (Choice1E.center) --(17End.center);
		\draw (Choice2L.center) --(Choice2.center);
		
		\draw (Choice2.center) --(21Start.center);
		\draw[matched] (Choice2.center) --(21Start.center);
		\draw (Choice2.center) --(22Start.center);
		\draw (Choice2ER.center) --(Choice2E.center);
		
		\draw (Choice2E.center) --(21End.center);
		\draw[matched] (Choice2E.center) --(21End.center);
		\draw (Choice2E.center) --(22End.center);
		\draw [draw=black, dashed] ([shift={(0, 1.75)}]Choice1Middle) rectangle ([shift={(0, -1)}]Choice2MiddleE);
			\draw [draw=white] ([shift={(-5.6, -1)}]Choice2MiddleE) -- ([shift={(0, -1)}]Choice2MiddleE);
		\end{pgfonlayer}

	\end{scope}
	\node[xshift=-.2cm, yshift=-1.1cm] at (boxB.west) {\textbf{(b)}};
\end{tikzpicture}

\end{center}

\caption{
Gadgets for the $\W[1]$-hardness result for \probl{RecovAP} and parameter $k$:
(a)
Selection gadget for row values (analogously for column values).
The depicted matching fixes value 2.
(b)
Component for one row in $S_{i, j}$ (analogously for one column). Each middle edge represents one tuple in $S_{i, j}$ This component exists for both columns and rows and the middle edges are identified if the corresponding tuples are the same.
}
\label{fig:rowselect}
\label{fig:gridcellgadget}
\end{figure}
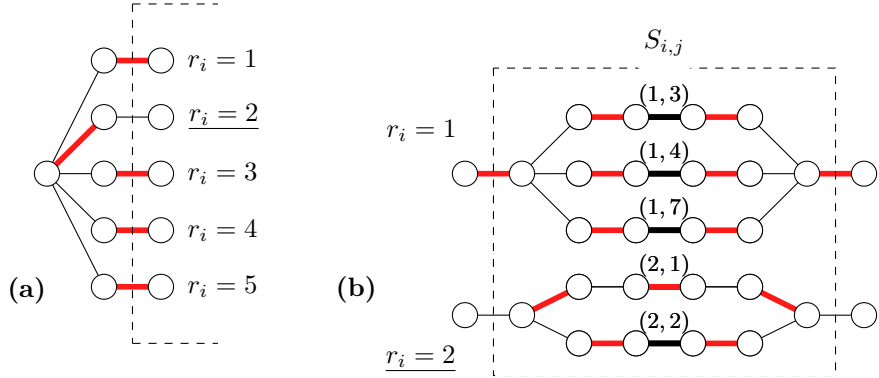

\begin{theorem}
\label{lemma:standard:hardness}
	\probl{RecovAP} is $\W[1]$-hard for parameter $k$ with edge costs $(c_1(e), c_2(e)) \in \set{(0, 0), (0, 1), (1,0), (1,1)}$ for all edges $e$, and unless \textup{ETH} fails it has no $f(k) n^{o(\sqrt{k})}$-time algorithm.
\end{theorem}

\begin{proof}
	Let $\mathcal{I} = (\mathcal{S}, n, \ell)$ be a \probl{Grid Tiling} instance.
	We construct a \probl{RecovAP} instance that asks for matchings of cost 0
		and with at least $k = \ell^2$ double edges $e \in M_1 \cap M_2$.
	We have an edge of cost $(0,0)$ for every tuple of every grid cell.
	Our construction forces exactly one $(0,0)$ cost edge in $M_1\cap M_2$ per grid cell,
		which marks that $(r_i,c_j) \in S_{i,j}$.
	We force these tuples to comply with a global selection of row and column values $r_1,\dots,r_\ell$ and $c_1,\dots,c_\ell$.
	
	Beside these edges for tuples, no edge has cost $(0,0)$.
	The edges defined in the following gadget used for fixing a row value have cost $(0,1)$,
		hence are cheap for $M_1$;
		vice versa the analogous edges used in the gadget for fixing the column values have cost $(1,0)$.
	Finally, additional edges to make the graph complete bipartite have cost $(1,1)$,
		which thus may never be selected by $M_1$ nor $M_2$.
	Thus we may note that for a non-complete bipartite graph edge costs $\{(0,0),(0,1),(1,0)\}$ suffice.
	
	Matching $M_1$ fixes the selection of the row values, while $M_2$ analogously fixes the selection for the column values.
	We use a selection gadget to fix a selection of a row value $r_i$ between 1 and $n$, see \cref{fig:rowselect}(a).
	The matching of the high degree vertex to a neighbor $u$ fixes the value of $r_i$.
	Neighbor $u$ is the only neighbor of $r_i$ that is not matched to its corresponding right neighbor $u'$.
	Hence, the selection of $r_i$ is propagated to the edges corresponding to tuples of a grid cell.
	All edges for the row selection gadget have cost $(0,1)$, hence are cheap for matching $M_1$.
	Later, we will introduce cheap edges for $M_2$ incident to the vertices of this selection gadget to guarantee that $M_2$ is a perfect matching.
	
	We extend the path from a selection of say $r_i=2$ to the $(0,0)$ cost edges corresponding to tuples of the first grid cell (of column 1). We then continue this path for every following grid cell (of column 2 to $\ell$), 
	see \cref{fig:gridcellgadget}(b) and how it can be connected cell by cell.
	We split the path of say selection $r_i=a$ at node $v$ into as many tuples with row value $a$ as there are in the corresponding grid cell.
	If the selection was $r_i=2$, then (by having the right path lengths) matching $M_2$ has to have $v$ matched to on of its right neighbors.
	The matching of $v$ to one of its neighbors, on a path to an edge corresponding to a tuple $(r_i=a,c_j=b)\in S_{i,j}$, corresponds to the fact that this tuple satisfy the condition for that cell.
		Note that for the not selected row values, in this example $r_i \neq 1$,
		none of the tuple edges will be matched.
	
	Analogously, the column values are selected and have paths extending to the $(0,0)$-cost edges.
	For $M_1,M_2$ to coincide at a $(0,0)$ cost edge of a grid cell $(i,j)$, here for tuple $(a,b)$, the selection gadget for $c_1$ has to select $b$ and analogously hast to extend its alternating path to the $(0,0)$ cost edge for the tuple $(a,b)$.
	Thus the tuple selection for a grid cell complies with the row and column value selection.
	
	For the construction, it remains to add the before mentioned further additional edges such that $M_1$ and $M_2$ are perfect matchings.
	The matchings can be completed by adding additional helper vertices:
		for every vertex $u_r$ of a row gadget a vertex $v$
		with a $(0,1)$ cost edge $(u_r,v)$.
	We also have a $(1,0)$ cost edge from $u$ to the analogous vertex $u_c$ that occurs in a column gadget.
	This way, the new vertex $v$ is matched by $M_1$ and $M_2$.
		These vertices and edges are depicted in \cref{fig:cellgadgetcomplete}(a) as the edges and vertices in faint color.
	
	\def\boxSize{2.5cm}
	\def\sideSep{0.2cm}
	\def\sideDist{0.25cm}
	
\begin{figure}[h]
\begin{center}
\captionsetup{position=b}
\begin{tikzpicture}[scale=0.8, every node/.style={draw, circle, fill=white, minimum size=.25cm,inner sep=0pt}]
\begin{scope}[local bounding box=boxA]

\coordinate (tl);
			\coordinate[right=\boxSize of tl] (tr);
			\coordinate[below=\boxSize of tl] (bl);
			\coordinate[right=\boxSize of bl] (br);
			
			\draw [dashed] (tl) -- (bl) -- (br) -- (tr) -- (tl);
			
			\node[below left=\boxSize/2 and 0.2cm of tl] (l1) {};
			\node[below =\sideSep of l1] (l2) {};
			\node[below =\sideSep of l2] (l3) {};
			
			\node[above =\sideSep of l1] (l4) {};
			\node[above =\sideSep of l4] (l5) {};
			
			\node[below right=\boxSize/2 and 0.2cm of tr] (r1) {};
			\node[below =\sideSep of r1] (r2) {};
			\node[below =\sideSep of r2] (r3) {};
			
			\node[above =\sideSep of r1] (r4) {};
			\node[above =\sideSep of r4] (r5) {};
			
			\node[above right=0.2cm and \boxSize/2 of tl] (t1) {};
			\node[right =\sideSep of t1] (t2) {};
			\node[right =\sideSep of t2] (t3) {};
			
			\node[left =\sideSep of t1] (t4) {};
			\node[left =\sideSep of t4] (t5) {};
			
			\node[below right=0.2cm and \boxSize/2 of bl] (b1) {};
			\node[right =\sideSep of b1] (b2) {};
			\node[right =\sideSep of b2] (b3) {};
			
			\node[left =\sideSep of b1] (b4) {};
			\node[left =\sideSep of b4] (b5) {};
			
			\node[left=\sideDist of l4] (lefttop) {};
			
			\node[draw=none, opacity=0] (bottomlefthelper) at ($(l2)!0.5!(l3)$) {};
			\node[left=\sideDist of bottomlefthelper] (leftbottom) {};
			
			\node[right=\sideDist of r4] (righttop) {};
			
			\node[draw=none, opacity=0] (bottomlefthelper) at ($(r2)!0.5!(r3)$) {};
			\node[right=\sideDist of bottomlefthelper] (rightbottom) {};
			
			\node[draw=none, opacity=0] (toplefthelper) at ($(t1)!0.5!(t4)$) {};
			\node[above=\sideDist of toplefthelper] (topleft) {};
			
			\node[above=\sideDist of t3] (topright) {};
			
			\node[draw=none, opacity=0] (bottomlefthelper) at ($(b1)!0.5!(b4)$) {};
			\node[below=\sideDist of bottomlefthelper] (bottomleft) {};
			
			\node[below=\sideDist of b3] (bottomright) {};
			
			\node[draw] (tuple1Start) at (0.66*\boxSize, -0.33*\boxSize) {};
			\node[draw, below right=0.04*\boxSize of tuple1Start] (tuple1End) {};
			
			\node[draw] (tuple2Start) at (0.33*\boxSize, -0.66*\boxSize) {};
			\node[draw, below right=0.04*\boxSize of tuple2Start] (tuple2End) {};
			
			\node[above left=0.2cm of tl, draw, opacity=0] (e1) {};
			\node[above left=0.2cm of e1, draw, opacity=0] (e2) {};
			\node[above left=0.2cm of e2, draw, opacity=0] (e3) {};
			\node[above left=0.15cm of e3, draw, opacity=0, gray!50] (e4) {};
			\node[above left=0.15cm of e4, draw, opacity=0] (e5) {};
			\node[above left=0.2cm of tl, draw, gray!50, fill=white] {};
			\node[above left=0.2cm of e1, draw, gray!50, fill=white] {};
			\node[above left=0.2cm of e2, draw, gray!50, fill=white] {};
			\node[above left=0.15cm of e3, draw, opacity=0, gray!50, fill=white] {};
			\node[above left=0.15cm of e4, draw, gray!50, fill=white] {};
			
			\node[below right=0.2cm of br, draw, opacity=0] (eb1) {};
			\node[below right=0.2cm of eb1, draw, opacity=0] (eb2) {};
			\node[below right=0.2cm of eb2, draw, opacity=0] (eb3) {};
			\node[below right=0.15cm of eb3, draw, opacity=0, gray!50] (eb4) {};
			\node[below right=0.15cm of eb4, draw, opacity=0] (eb5) {};
			\node[below right=0.2cm of br, draw, gray!50, fill=white] {};
			\node[below right=0.2cm of eb1, draw, gray!50, fill=white] {};
			\node[below right=0.2cm of eb2, draw, gray!50, fill=white] {};
			\node[below right=0.15cm of eb3, draw, opacity=0, gray!50, fill=white] {};
			\node[below right=0.15cm of eb4, draw, gray!50, fill=white] {};
			
			\begin{pgfonlayer}{bg}
				\draw[matchedalt, opacity=0.3] (e1.center) -- (l5);
				\draw[matched, opacity=0.3] (e1.center) -- (t5);
				\draw[matchedalt, opacity=0.3] (e2.center) to[bend right=10] (l4.center);
				\draw[matched, opacity=0.3] (e2.center) to[bend left=10] (t4.center);
				\draw[matchedalt, opacity=0.3] (e3.center) to[bend right=5] (l1.center);
				\draw[matched, opacity=0.3] (e3.center) to[bend left=5] (t1.center);
				\draw[matchedalt, opacity=0.3] (e5.center) to[bend right=5] (leftbottom.center);
				\draw[matched, opacity=0.3] (e5.center) to[bend left=5] (topright.center);
				
				\draw[loosely dotted, opacity=0.7] (e3.north west) -- (e5.south east);
				
				\draw[matchedalt, opacity=0.3] (eb1.center) -- (r3);
				\draw[matched, opacity=0.3] (eb1.center) -- (b3);
				\draw[matchedalt, opacity=0.3] (eb2.center) to[bend right=10] (r2.center);
				\draw[matched, opacity=0.3] (eb2.center) to[bend left=10] (b2.center);
				\draw[matchedalt, opacity=0.3] (eb3.center) to[bend right=5] (r1.center);
				\draw[matched, opacity=0.3] (eb3.center) to[bend left=5] (b1.center);
				\draw[matchedalt, opacity=0.3] (eb5.center) to[bend right=5] (righttop.center);
				\draw[matched, opacity=0.3] (eb5.center) to[bend left=5] (bottomleft.center);
				
				\draw[loosely dotted, opacity=0.7] (eb5.north west) -- (eb3.south east);
				
				\draw (lefttop.center) -- (l5.center);
				\draw[matched] (lefttop.center) -- (l4.center);
				\draw (lefttop.center) -- (l1.center);
				\draw (leftbottom.center) -- (l2.center);
				\draw (leftbottom.center) -- (l3.center);
				\draw (righttop.center) -- (r5.center);
				\draw[matched] (righttop.center) -- (r4.center);
				\draw (righttop.center) -- (r1.center);
				\draw (rightbottom.center) -- (r2.center);
				\draw (rightbottom.center) -- (r3.center);
				\draw (topleft.center) -- (t5.center);
				\draw (topleft.center) -- (t4.center);
				\draw (topleft.center) -- (t1.center);
				\draw[matchedalt] (topleft.center) -- (t2.center);
				\draw (topright.center) -- (t3.center);
				\draw (bottomleft.center) -- (b5.center);
				\draw (bottomleft.center) -- (b4.center);
				\draw (bottomleft.center) -- (b1.center);
				\draw[matchedalt] (bottomleft.center) -- (b2.center);
				\draw (bottomright.center) -- (b3.center);
				
				\draw (l4.center) -- (tuple1Start.center);
				\draw (t2.center) -- (tuple1Start.center);
				\draw (r4.center) -- (tuple1End.center);
				\draw (b2.center) -- (tuple1End.center);
				\draw[transform canvas={xshift=-.1em, yshift=-.1em}, matchedalt] (tuple1Start.center) -- (tuple1End.center);
				\draw[transform canvas={xshift=.1em, yshift=.1em}, matched] (tuple1Start.center) -- (tuple1End.center);
				\draw[matched] (l2.center) -- (tuple2Start.center);
				\draw[matchedalt] (t4.center) -- (tuple2Start.center);
				\draw[transform canvas={xshift=-.1em, yshift=-.1em}, matchedalt] (tuple1Start.center) -- (tuple1End.center);
				\draw[transform canvas={xshift=.1em, yshift=.1em}, matched] (tuple1Start.center) -- (tuple1End.center);
				\draw[matched] (r2.center) -- (tuple2End.center);
				\draw[matchedalt] (b4.center) -- (tuple2End.center);
				\draw (tuple2Start.center) -- (tuple2End.center);
			\end{pgfonlayer}

	\end{scope}
	\node[xshift=-.1cm, yshift=-1.75cm, draw=none] at (boxA.west) {\textbf{(a)}};
	\begin{scope}[local bounding box=boxB,xshift=7cm,yshift=-2cm]

	\node [label=below:$l$] (l) {};

	\node[right of=l] (m) {};
	\node[right of=m] (h) {};
	\node[right of=h, label=below:$r$] (r) {};
	
	\node[above of=m] (v) {};
	\node[above of=v, label=right:$t$] (t) {};
	\node[below of=m, label=right:$b$] (b) {};
	\node[right of=v] (e) {};
	\begin{pgfonlayer}{bg}
%
	
	\draw[matched] (l.center) -- (m.center);
	\draw[matched, dashed] (m.center) -- (h.center);
	\draw[matched] (h.center) -- (r.center);
	\draw[matched] (v.center) -- (e.center);
	
	\draw[matchedalt, dashed] (b.center) -- (m.center);
	\draw[matchedalt] (m.center) -- (v.center);
	\draw[matchedalt, dashed] (v) -- (t);
	\draw[matchedalt] (h.center) -- (e.center);
	\end{pgfonlayer}

	\end{scope}
	\node[xshift=-.2cm, yshift=-1.5cm, draw=none] at (boxB.west) {\textbf{(b)}};
\end{tikzpicture}
\end{center}
\caption{
(a)
A sketch of the surrounding of a grid cell component (in the dashed box) with matching $M_1$ in green and $M_2$ in red.
Most of the edges leading into the box are left out, except for two example diagonal $(0,0)$ cost edges, one in both matchings, one in none.
Outside in faint color are the additional helper vertices to make the matchings perfect.
By positioning these helper vertices on a diagonal as depicted, there are only crossings of edges of cost $(0,1)$ and $(1,0)$.
(b)
A crossing gadget that replaces a crossing of a $(0,1)$ cost edge $(l,r)$ and $(1,0)$ cost edge $(t,b)$.
The red (dashed and fully drawn) edges have cost $(0,1)$ and the green cost $(1,0)$.
The fully drawn red edges show a matching replacing $(\ell,r)\in M_1$
	while the fully drawn green edges show a matching replacing $(t,b)\notin M_2$.
}
\label{fig:planarGadget}
\label{fig:cellgadgetcomplete}
\end{figure}
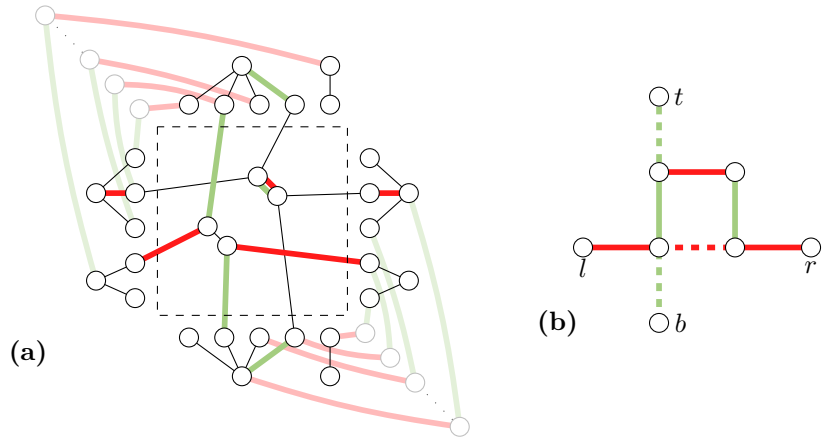
		
	
	For the forward direction, of the correctness,
		assume that there exists a solution to the \probl{Grid Tiling} instance.
	Then the matchings $M_1$ and $M_2$ can be chosen as described in the construction,
		which yields a solution with cost 0.
	
	For the reverse direction, assume that there is a solution to the constructed \probl{RecovAP} instance,
		in other words matchings $M_1$ and $M_2$ of intersection $\ell^2$ and of cost 0.
	Then only the $(0,0)$ cost edges corresponding to the tuples in the grid cells can be in the intersection $M_1 \cap M_2$.
	By construction every selection gadget for a grid cell, corresponding to one specific row and column, there is at most one edge where $M_1$ and $M_2$ coincide.
	Because the number of coinciding edges must be at least $\ell^2$,
		any solution selects exactly one tuple for every cell.
	Further, as noted before, the row value $r_i$ for all the gird cells of one row $i$ has to be the same,
		namely as fixed at the row selection gadget; analogously each column $j$ has a fixed value $c_j$.
	Then these row and column values $r_1,\dots,r_\ell$ and $c_1,\dots,c_\ell$ are a solution to the original \probl{Grid Tiling} instance.
	
	To show the ETH lower bound assume for contradiction that there is an algorithm with running time $f(k) n^{o(\sqrt{k})}$. Then an instance of Grid Tiling can be transformed in polynomial time into an instance of \probl{RecovAP}. For the parameter it holds that $k = \ell^2$. So this leads to a running time $f(\ell) n^{o(\sqrt{\ell^2})} = f(\ell) n^{o(\ell)}$. This is a contradiction.
\end{proof}

This hardness results also translates to planar graphs with the help of a crossing gadget.
Since the input graph is not complete, we no longer need edge costs $(1,1)$.
The key observation is that by an arrangement of the nodes in the plane as in \cref{fig:planarGadget}(a), there are only crossing edges of cost $(0,1)$ and $(1,0)$.
These are from horizontal edges for $M_1$ fixing a row value and vertical edges fixing a column value $M_2$, as well as from the additional helper vertices.
We replace each pair of crossing edges $(\ell,r), (t,b)$ with a construction as depicted in \cref{fig:planarGadget}(b), for which we leave to the reader to verify its correctness.
It is important that the four new vertices may always be matched within this crossing gadget,
	and thus no further helper vertices are needed.
Further, note that this construction can be easily chained in order to suit cases where $(\ell,r)$ or $(t,b)$ cross more than one other edge.

	\begin{corollary}
		\probl{RecovAP} is $\textup{W}[1]$-hard on planar graphs for parameter $k$ with $(c_1(e), c_2(e)) \in \set{(0, 0), (0, 1), (1,0)}$ for all edges $e$, and unless \textup{ETH} fails it has no $f(k) n^{o(\sqrt{k})}$-time algorithm.
	\end{corollary}

\begin{figure}[t]
	\centering
	\begin{tikzpicture}[every node/.style={draw,circle,fill=white}]
		\node [label=$l$] (t1) at (0, 0) {};
		\node (t2) at (1, 0) {};
		\node (t3) at (2, 0) {};
		\node (t4) at (3, 0) {};
		\node (t5) at (4, 0) {};
		\node  (t6) at (5, 0) {};
		\node (t7) at (6, 0) {};
		\node [label=$r$] (t8) at (7, 0) {};
		
		\node (a1) at (1, 1) {};
		\node (a2) at (1, 2) {};
		\node (a3) at (1, 3) {};
		
		\node (b1) at (2, 1) {};
		\node (b2) at (2, 2) {};
		\node (b3) at (2, 3) {};
		
		\node (c1) at (3, 1) {};
		\node (c2) at (3, 2) {};
		\node (c3) at (3, 3) {};
		
		\node (d1) at (4, 1) {};
		\node (d2) at (4, 2) {};
		\node (d3) at (4, 3) {};
		
		\node (e1) at (5, 1) {};
		\node (e2) at (5, 2) {};
		\node (e3) at (5, 3) {};
		
		\node (f1) at (6, 1) {};
		\node (f2) at (6, 2) {};
		\node (f3) at (6, 3) {};
		
		\draw (t2) -- (t3);
		\draw (t4) -- (t5);
		\draw (t6) -- (t7);

		\draw[matched] (t1) -- (t2);
		\draw[matched] (t3) -- (t4);
		\draw[matched] (t5) -- (t6);
		\draw[matched] (t7) -- (t8);
		
		\draw[matched] (a1) -- (a2);
		\draw[matched] (b1) -- (b2);
		\draw[matched] (c1) -- (c2);
		\draw[matched] (d1) -- (d2);
		\draw[matched] (e1) -- (e2);
		\draw[matched] (f1) -- (f2);
		
		\draw[matched] (a3) -- (b3);
		\draw[matched] (c3) -- (d3);
		\draw[matched] (e3) -- (f3);
		
		\draw[matchedalt] (t2) -- (a1);
		\draw[matchedalt] (t3) -- (b1);
		\draw[matchedalt] (t4) -- (c1);
		\draw[matchedalt] (t5) -- (d1);
		\draw[matchedalt] (t6) -- (e1);
		\draw[matchedalt] (t7) -- (f1);
		
		\draw[matchedalt] (a2) -- (a3);
		\draw[matchedalt] (b2) -- (b3);
		\draw[matchedalt] (c2) -- (c3);
		\draw[matchedalt] (d2) -- (d3);
		\draw[matchedalt] (e2) -- (e3);
		\draw[matchedalt] (f2) -- (f3);
		
	\end{tikzpicture}
	\caption{
	This construction replaces an edge $(l,r)$ and effectively simulates an edge of cost $(k + 1, 0)$,
		here for the sake of simplicity with $k=3$.
	It consists of one path of length $2k + 1$ where all edges have cost $(1, 0)$.
	Every second pair of neighboring vertices (exculding $l$ and $r$) is connected by another extra path of length $2k + 1$.
	We set the cost of the edges in such extra paths to $(0,1)$ and $(1,0)$,
		such that the depicted matchings have cost $0$.
	Any other choice for the matching edges on the top paths directly leads to a total cost that exceeds $k$.
	Hence, the main path can only contain matching edges of the red matching.}
	\label{fig:edgeReplacement}
\end{figure}
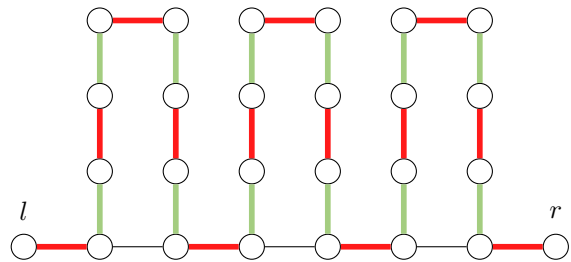

As a final step regarding planar graphs, we avoid vertices of degree $>4$, and we avoid $(0,0)$ cost edges,
	thus showing hardness for costs $(0,1)$ and $(1,0)$.
The key step is to replace a single $(0,1)$ edge by a long path-like gadget that effectively simulates a cost $(0,k+1)$ edge, analogously for a $(1,0)$ edge.
\cref{fig:edgeReplacement} shows such a gadget.
Hence we may allow to have costs of up to $k$ and it still holds that none of the original $(0,1)$ and $(1,0)$ edges may be matched by the matching having non-zero cost.
Then we may change the $(0,0)$ cost edges to cost $(0,1)$ (or equally we may change them to cost $(1,0)$).
Finally, the high degree vertices can be avoided by a simple binary tree construction.

    \begin{corollary}
    	\label{lemma:hardness:low:degree:2:weights}
    	\probl{RecovAP} is $\textup{W}[1]$-hard on planar graphs with maximum degree 4 for parameter $k$ with $(c_1(e), c_2(e)) \in \set{(0, 1), (1,0)}$ for all edges $e$, and unless \textup{ETH} fails it has no $f(k) n^{o(\sqrt{k})}$-time algorithm.
    	\label{col:RecovAP}
	\end{corollary}

In a similar way, we can also show W[1]-hardness for the dual parameter $k' = n-k$,
	hence the problem that asks to have all of the $n$ matching edges of $M_1$ and $M_2$ to coincide except for up to $k$ exceptions.
In robust optimization this parameter is of importance and called the recoverability parameter.
These results then complement the simple XP algorithms as mentioned in the introduction.

\begin{figure}
	\centering
		\centering
		\resizebox{.35\textwidth}{!} {
		\begin{tikzpicture}[every node/.style={draw,circle,fill=white}, smooth]
		\node (l) {};
		\node[right =0.88cm of l] (r) {};
		\node[left of=l] (l2) {};
		\node[left of=l2] (l3) {};
		\node[left of=l3] (l4) {};

		\node[right of=r] (r2) {};
		\node[right of=r2] (r3) {};
		\node[right of=r3] (r4) {};
		
		\node[above of=l] (t2) {};
		\node[above of=t2] (t3) {};
		\node[above of=t3] (t4) {};
		
		\node[below of=r] (b2) {};
		\node[below of=b2] (b3) {};
		\node[below of=b3] (b4) {};
		
		\node[above of=l2, fill=white] (tl) {};
		\node[below of=r2, fill=white] (br) {};
		
		
		\begin{pgfonlayer}{bg}
			\draw[transform canvas={yshift=.11em}, matched, dashed] (l.center) -- (r.center);
			\draw[transform canvas={yshift=-.11em}, matchedalt, dashed] (l.center) -- (r.center);
		\end{pgfonlayer}
		
		\draw[line width=0.1cm, dashed] (l4) -- (l3);
		\doublematchedH{2}{l3}{l2}
		\draw (l2) -- (l);
		
		\draw[line width=0.1cm, dashed] (r4) -- (r3);
		\doublematchedH{2}{r3}{r2}
		\draw (r2) -- (r);

		\draw[line width=0.1cm, dashed] (t4) -- (t3);
		\doublematchedV{2}{t3}{t2}
		\draw (t2) -- (l);
		
		\draw[line width=0.1cm, dashed] (b4) -- (b3);
		\doublematchedV{2}{b3}{b2}
		\draw (b2) -- (r);
%
%
		
		\coordinate (middle1) at ($(l2)!0.7!(l)$);
		\coordinate (middle2) at ($(r2)!0.7!(r)$);
		\begin{pgfonlayer}{bg}
		\draw[matched, transform canvas={xshift=.11em, yshift=.11em}] plot[tension=0.5] coordinates {(tl)([shift={(0, -0.4)}]middle1)([shift={(0, -0.55)}]middle2)(br)};
		\draw[matchedalt, transform canvas={xshift=-.11em, yshift=-.11em}] plot[tension=0.5] coordinates {(tl)([shift={(0, -0.4)}]middle1)([shift={(0, -0.55)}]middle2)(br)};
		
		\draw (l2) -- (l);
		\draw (r2) -- (r);
		\draw (t2) -- (l);
		\draw (b2) -- (r);
		\draw (l2) -- (tl);
		\draw (t2) -- (tl);
		\draw (r2) -- (br);
		\draw (b2) -- (br);
		
		\end{pgfonlayer}
		\end{tikzpicture}
	}
	\phantom{spa}	
		\resizebox{.35\textwidth}{!}{
		\begin{tikzpicture}[every node/.style={draw,circle,fill=white}, smooth]
		\node (l) {};
		\node[right =0.88cm of l] (r) {};
		\node[left of=l] (l2) {};
		\node[left of=l2] (l3) {};
		\node[left of=l3] (l4) {};

		\node[right of=r] (r2) {};
		\node[right of=r2] (r3) {};
		\node[right of=r3] (r4) {};
		
		\node[above of=l] (t2) {};
		\node[above of=t2] (t3) {};
		\node[above of=t3] (t4) {};
		
		\node[below of=r] (b2) {};
		\node[below of=b2] (b3) {};
		\node[below of=b3] (b4) {};
		
		\node[above of=l2, fill=white] (tl) {};
		\node[below of=r2, fill=white] (br) {};
		
		
		\begin{pgfonlayer}{bg}
			\draw[line width=0.1cm, dashed] (l) -- (r);
		\end{pgfonlayer}
		
		\begin{pgfonlayer}{bg}
			\draw[transform canvas={yshift=.11em}, matched, dashed] (l4.center) -- (l3.center);
			\draw[transform canvas={yshift=-.11em}, matchedalt, dashed] (l4.center) -- (l3.center);
		\end{pgfonlayer}
		\draw (l3) -- (l2);
		\draw[matchedalt] (l2) -- (l);
		
		\begin{pgfonlayer}{bg}
		\draw[transform canvas={yshift=.11em}, matched, dashed] (r4.center) -- (r3.center);
		\draw[transform canvas={yshift=-.11em}, matchedalt, dashed] (r4.center) -- (r3.center);
		\end{pgfonlayer}
		\draw (r3) -- (r2);
		\draw[matchedalt] (r2) -- (r);

		\begin{pgfonlayer}{bg}
		\draw[transform canvas={xshift=.11em}, matched, dashed] (t4.center) -- (t3.center);
		\draw[transform canvas={xshift=-.11em}, matchedalt, dashed] (t4.center) -- (t3.center);
		\end{pgfonlayer}
		\draw (t3) -- (t2);
		\draw[matched] (t2) -- (l);
		
		\begin{pgfonlayer}{bg}
		\draw[transform canvas={xshift=.11em}, matched, dashed] (b4.center) -- (b3.center);
		\draw[transform canvas={xshift=-.11em}, matchedalt, dashed] (b4.center) -- (b3.center);
		\end{pgfonlayer}
		\draw (b3) -- (b2);
		\draw[matched] (b2) -- (r);
		
		\draw[matched] (l2) -- (tl);
		\draw[matchedalt] (t2) -- (tl);
		
		\draw[matched] (r2) -- (br);
		\draw[matchedalt] (b2) -- (br);
		
		\coordinate (middle1) at ($(l2)!0.7!(l)$);
		\coordinate (middle2) at ($(r2)!0.7!(r)$);
		\begin{pgfonlayer}{bg}
		\draw plot[tension=0.5] coordinates {(tl)([shift={(0, -0.4)}]middle1)([shift={(0, -0.55)}]middle2)(br)};

		\draw (l2) -- (l);
		\draw (r2) -- (r);
		\draw (t2) -- (l);
		\draw (b2) -- (r);
		\draw (l2) -- (tl);
		\draw (t2) -- (tl);
		\draw (r2) -- (br);
		\draw (b2) -- (br);
		
		\end{pgfonlayer}
		\end{tikzpicture}
	}
	\caption{
	Construction for an edge corresponding to a tuple, as used for the hardness for parameter $n - k$.
	The dashed edges represent paths of length $4\ell^2 + 1$.
	Almost all edges have cost $(0,0)$.	
	The only exceptions are the edges in the two 4-cycles,
		which have cost $(0, 1)$ and $(1, 0)$ such that the cost for the matchings in the right figure is still $0$.}
	\label{fig:tupleGadgetkPrime}
\end{figure}
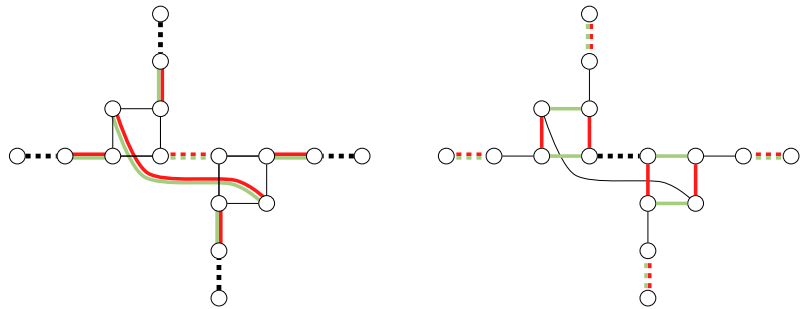

\begin{theorem}\label{lemma:revocap:hardness:recoverability}
	\probl{RecovAP} is $\W[1]$-hard for the recoverability parameter $k' = n-k$ with $(c_1(e), c_2(e)) \in \set{(0,0), (0, 1), (1,0), (1,1)}$ for all edges $e$, and unless \textup{ETH} fails it has no $f(k) n^{o(\sqrt{k})}$-time algorithm.
\end{theorem}
\begin{proof}
We adapt our hardness reduction from \probl{Grid Tiling} as seen in \cref{lemma:standard:hardness} to suit for the dual parameter $k' = n-k$.
Again, we set the allowed cost to 0.
Where before we had an edge $e$ exclusively in one matching, we change the construction such that $e \in M_1\cap M_2$,
	and vice versa replace an edge corresponding to a tuple by a construction that avoids 4 double edges in case that its corresponding row and column values are selected.
Thus we set $k'= 4\ell^2$.
We no longer need the additional helper vertices as introduced in the original hardness reduction.

We adapt the row selection gadget of the original construction, a star, by subdividing each leg $4\ell^2+1$ times, see \cref{fig:rowselectkprime}; analogously for the column selection.
All such edges have cost $(0,0)$.
If the matching edges of $M_1,M_2$ incident to a vertex $v$ on such paths do not coincide,
	then they do not coincide on two paths of length $4\ell^2 + 1 > k'$.
Thus $M_1,M_2$ have to agree on a common row and column selection gadget.

We replace each edge from the original construction corresponding to a tuple of a grid cell by a gadget as shown in \cref{fig:tupleGadgetkPrime}.
Here, we again use paths of length $4\ell^2+1$, shown as dashed edges in the figure, such that on the edges of such a path the matchings $M_1$ and $M_2$ have to coincide.
Almost all edge costs are $(0,0)$ except for the edges of the two 4-cycles.
There we assign costs $(0,1)$ and $(1,0)$ such that the matchings $M_1$ and $M_2$ as shown in the right of \cref{fig:tupleGadgetkPrime} have cost 0.
It is easy to see that there are only two possible local solutions for this gadget:
Either $M_1,M_2$ are as shown in the left part of the figure,
	which corresponds to a tuple without the selected row and column values;
Or $M_1,M_2$ are as shown in the right part of the figure,
	which corresponds to a tuple $(r_i,c_j)$ with the selected row and column value,
	however with 4 matching edges for each matching that do not coincide with the other matching.
Thus any solution may only have at most $\ell^2$ grid cells with a selected tuple.
Since every grid cell has to have at least one tuple selected, exactly one is chosen.

We add the remaining edges to make the graph complete bipartite with cost $(1,1)$.
Since the allowed cost is 0, they can never be selected by any matching.
\end{proof}

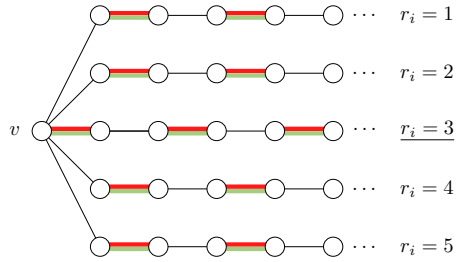
\begin{figure}
	\centering
		\centering
		\resizebox{.45\textwidth}{!}{
		\begin{tikzpicture}[every node/.style={draw,circle,fill=white}]
		
		\node (n1S) {};
		\node[below of=n1S] (n2S) {};
		\node[below of=n2S] (n3S) {};
		\node[below of=n3S] (n4S) {};
		\node[below of=n4S] (n5S) {};
		
		\node[right of=n1S] (n1) {};
		\node[right of=n2S] (n2) {};
		\node[right of=n3S] (n3) {};
		\node[right of=n4S] (n4) {};
		\node[right of=n5S] (n5) {};
	
		\node[right of=n1] (n1S2) {};
		\node[below of=n1S2] (n2S2) {};
		\node[below of=n2S2] (n3S2) {};
		\node[below of=n3S2] (n4S2) {};
		\node[below of=n4S2] (n5S2) {};
		
		\node[right of=n1S2] (n1S3) {};
		\node[below of=n1S3] (n2S3) {};
		\node[below of=n2S3] (n3S3) {};
		\node[below of=n3S3] (n4S3) {};
		\node[below of=n4S3] (n5S3) {};
			
		\node[right of=n1S3, label={right: $\cdots \quad r_i = 1$}] (n14) {};
		\node[right of=n2S3, label={right: $\cdots \quad r_i = 2$}] (n24) {};
		\node[right of=n3S3, label={right: $\cdots \quad \underline{r_i = 3}$}] (n34) {};
		\node[right of=n4S3, label={right: $\cdots \quad r_i = 4$}] (n44) {};
		\node[right of=n5S3, label={right: $\cdots \quad r_i = 5$}] (n54) {};
		
		\node[left of=n3S, label={left:$v$}] (start) {};
		
		\begin{pgfonlayer}{bg}    
		\draw (start.center) --(n1S.center);
		\draw (start.center) --(n2S.center);
		\draw (n3S.center) --(n3.center);
		\draw (start.center) --(n4S.center);
		\draw (start.center) --(n5S.center);
		
		\draw (n1.center) --(n1S2.center);
		\draw (n2.center) --(n2S2.center);
		\draw (n3S.center) --(n3S3.center);
		\draw (n4.center) --(n4S2.center);
		\draw (n5.center) --(n5S2.center);

		\draw (n1S3.center) --(n14.center);
		\draw (n2S3.center) --(n24.center);
		\draw (n4S3.center) --(n44.center);
		\draw (n5S3.center) --(n54.center);

		\doublematchedH{2-}{start}{n3S}
		
		\doublematchedH{3-}{n1S}{n1}
		\doublematchedH{3-}{n2S}{n2}
		\doublematchedH{2-}{n3}{n3S2}
		\doublematchedH{3-}{n4S}{n4}
		\doublematchedH{3-}{n5S}{n5}
		
		\doublematchedH{3-}{n1S2}{n1S3}
		\doublematchedH{3-}{n2S2}{n2S3}
		\doublematchedH{2-}{n3S3}{n34}
		\doublematchedH{3-}{n4S2}{n4S3}
		\doublematchedH{3-}{n5S2}{n5S3}
		
		\end{pgfonlayer}
		\end{tikzpicture}
		}
		\caption{Row selection gadget for the dual parameter $n - k$ (analogously for column selection).
		We reuse the selection gadget for the natural parameter $k$,
			but subdivide the legs of the star $4\ell^2 + 1$ times.
		All edges have cost $(0,0)$.
		If the matching edges of $M_1,M_2$ incident to a vertex $v$ do not coincide,
			then they do not coincide on two paths of length $4\ell^2 + 1 > k'$.
		Thus $M_1,M_2$ agree on a common row and column selection gadget.
		}
		\label{fig:rowselectkprime}
\end{figure}

Ideally, \autoref{lemma:revocap:hardness:recoverability} would translate to planar graphs
	by using a crossing gadget, analogously as for the natural parameter $k$.
However, according to Gurjar et al.~\cite{gurjar2012planarizing}, such a crossing gadget does not exist.

\medskip

On the contrary to planar graphs, we consider graphs of bounded treewidth.
We give a fixed parameter tractable algorithm in the treewidth of the input graph (without the intersection size $k$ as parameter).
Our algorithm is based on dynamic programming over the tree decomposition, and we refer to \cref{sec:tw} for the details.
We are not aware of any Courcelle-like meta-theorem which, in addition to minimizing the cost, allows to model a constraint $|M_1 \cap M_2| \geq k$.

\begin{theorem}\label{thm:tw}
	\probl{RecovAP} is in \textup{FPT} with respect to the treewidth of the input graph.
\end{theorem}

\section{Monge and Anti-Monge Matrices}\label{sec:monge}
\newcommand{\cyclic}{cyclic\xspace}

%
%

In this section we develop a polynomial time algorithm for the special case of 
\probl{RecovAP} if the cost function $c_1$ is given by a Monge matrix $A = (a_{i,j}) \in \rr^{n \times n}$ and the 
cost function of $c_2$ is given by an Anti-Monge matrix $ B = (b_{i,j}) \in \rr^{n \times n}$.
The matrix $A$ is called a Monge matrix if for all $i < k$ and $j < l$ it holds that
	$a_{i,j} + a_{k,l} \leq a_{i,l} + a_{k,j}$. Analogously, the matrix $B$ is called an Anti-Monge matrix 
if for all $i < k$ and $j < l$ it holds that 
$b_{i,j} + b_{k,l} \geq b_{i,l} + b_{k,j}$.
Let $U = \{u_1, \dots, u_n\}$ and $V = \{v_1, \dots, v_n\}$ be the bipartition of the vertex 
set of $K_{n,n}$. Then the cost of edge $\{u_i, v_j\}$ is given by $c_1(\{u_i, u_j\}) = a_{i,j}$
and $c_2(\{u_i, v_j\}) = b_{i,j}$. 

Note, that it is a well-known 
result that if the cost function of the \probl{AP} is given by a Monge matrix then the diagonal $\{ \{u_i, v_i\} \colon i=1, \dots, n \}$ is an optimal solution, similarly the anti-diagonal for Anti-Monge matrices.

Surprisingly, for these special cost functions also the \probl{RecovAP} has an optimal solution of a similar combinatorial structure, as we show in the following.
We illustrate solutions in matrix form by highlighting entry 
$(i,j)$ with a square if $\{u_i,v_j\} \in M_1$ and with a cross if $\{u_i,v_j\} \in M_2$.
See \cref{fig:monge1} for an illustration of the combinatorial structure of optimal solutions 
as claimed in \cref{thm:mongestructure}.

\begin{theorem}\label{thm:mongestructure}
	Let $c_1$ be given by a Monge matrix $A$ and $c_2$ be given by an Anti-Monge matrix $B$.
	Then, there exists an optimal solution $M_1, M_2$ to the \probl{RecovAP} such that
	\begin{enumerate}
		\item $\{u_i, v_i\} \in M_1$ for all $i=1,\dots,\lfloor \frac{n-k}{2} \rfloor$ and $i= \lceil \frac{n+k}{2} \rceil, \dots, n$,
		\item $\{u_i, v_{n+1-i}\} \in M_2$ for all $i=1,\dots,\lfloor \frac{n-k}{2} \rfloor$ and $i= \lceil \frac{n+k}{2} \rceil, \dots, n$,
		\item $M_1 \cap M_2$ is an optimal solution to the \probl{AP} with cost $c_1+c_2$ on the complete bipartite subgraph induced by the sets of vertices 
		$\{u_i \mid i = \lfloor \frac{n-k}{2} \rfloor+1, \dots, \lceil \frac{n+k}{2} \rceil-1\}$ and 
		$\{v_i \mid i = \lfloor \frac{n-k}{2} \rfloor+1, \dots, \lceil \frac{n+k}{2} \rceil-1\}$. 
	\end{enumerate}
\end{theorem}

\newcommand{\herescale}{3.6}
\newcommand{\hereratio}{0.45625}

\begin{figure}[t]
\captionsetup{position=b}
\begin{tikzpicture}[scale=1]
\begin{scope}[local bounding box=boxA,scale={\herescale}, every node/.style={draw, RWTH_green100, minimum width=5pt, minimum height=5pt}, smooth]

			\coordinate (bl) at (0,0);
			\coordinate (tl) at (1,0);
			\coordinate (br) at (0,1);
			\coordinate (tr) at (1,1);

			\node[draw=none,RWTH_red100] at (0.05, 0.05) {$\times$};
			\node[draw=none,RWTH_red100] at (0.11, 0.11) {$\times$};
			\node[draw=none,RWTH_red100] at (0.17, 0.17) {$\times$};
			\node[draw=none,RWTH_red100] at (0.83, 0.83) {$\times$};
			\node[draw=none,RWTH_red100] at (0.89, 0.89) {$\times$};
			\node[draw=none,RWTH_red100] at (0.95, 0.95) {$\times$};

			\node[] at (0.05, 0.95) {};
			\node[] at (0.11, 0.89) {};
			\node[] at (0.17, 0.83) {};
			\node[] at (0.95, 0.05) {};
			\node[] at (0.89, 0.11) {};
			\node[] at (0.83, 0.17) {};

			\draw[line width=.25mm, white] (bl.south west) -- (br.south east) -- (tr.north east) -- (tl.north west) --  (bl.south west);
			\draw[line width=.25mm] (bl.south west) -- (br.south east) -- (tr.north east) -- (tl.north west) --  (bl.south west);
			
			\node[draw=none, outer sep=0.05cm] (innertl) at (0.24,0.76) {};
			\node[outer sep=0.05cm] (innerbr) at (0.76,0.24) {};
			\node[draw=none, inner sep=-0.0cm, outer sep=0.05cm, RWTH_red100] (innerbl) at (0.24,0.24) {$\times$};
			\node[draw=none, inner sep=-0.0cm, outer sep=0.05cm, RWTH_red100] (innertr) at (0.76,0.76) {$\times$};
			

			\draw[RWTH_tuerkis100, line width=.25mm] (innerbl.south west) -- (innerbr.south east) -- (innertr.north east) -- (innertl.north west) --  (innerbl.south west);

			\node[] (innertlalt) at (0.24,0.6) {};
			\node[draw=none, inner sep=-0.0cm, outer sep=0.05cm, RWTH_red100] at (0.65,0.6) {$\times$};
			\node[] at (0.4,0.3) {};
			\node[draw=none, inner sep=-0.0cm, outer sep=0.05cm, RWTH_red100] at (0.4,0.3) {$\times$};
			\node[] (inneralt2) at (0.5,0.76) {};
			\node[draw=none, inner sep=-0.0cm, outer sep=0.05cm, RWTH_red100]  at (0.5,0.4) {$\times$};
			\node[]  at (0.65,0.4) {};
			\node[draw=none] (innerbr) at (0.5,0.6) {};
			
			\draw[RWTH_orange100, ->, thick] (innertlalt) -- (innertl.center);
			\draw[RWTH_orange100, ->, thick] (inneralt2) -- (innerbr);

	\end{scope}
	\node[xshift=-.25cm, yshift=-1.8cm] at (boxA.west) {\textbf{(a)}};
	\begin{scope}[local bounding box=boxB,xshift=4.5cm, yshift=\herescale cm, scale={\herescale*\hereratio}, every node/.style={draw, RWTH_green100, minimum width=5pt, minimum height=5pt}, smooth]

		\draw[line width=.25mm, white] (0,0) -- (2.2,0) -- (2.2,-2.2) -- (0,-2.2) -- cycle;
		\draw[line width=.25mm] (0,0) -- (2.2,0) -- (2.2,-2.2) -- (0,-2.2) -- cycle;
		\node[] at (2,-1) {};
		\node[draw=none, RWTH_red100] at (2,-1) {$\times$};
		
		\node[] at (1.5,-1.5) {};
		\node[] at (0.5,-0.5) {};
		
		\node[draw=none, RWTH_red100] at (1.5,-0.5) {$\times$};
		\node[draw=none, RWTH_red100] at (0.5,-1.5) {$\times$};

		

		\draw[RWTH_orange100, ->, thick] (2,-1) -- (2,-0.5);
		\draw[RWTH_orange100, ->, thick] (1.5,-0.5) -- (1.5,-1);
		
		\draw[RWTH_tuerkis100, ->, thick] (2,-1) -- (2,-1.5);
		\draw[RWTH_tuerkis100, ->, thick] (1.5,-1.5) -- (1.5,-1);

	\end{scope}
	\node[xshift=-.35cm, yshift=-1.8cm] at (boxB.west) {\textbf{(b)}};
	\begin{scope}[local bounding box=boxC, yshift=\herescale cm, xshift=9cm, scale=\herescale*\hereratio, every node/.style={draw, line width=.8pt, RWTH_green100, minimum width=5pt, minimum height=5pt}, smooth]

		\draw[line width=.25mm, white] (0,0) -- (2.2,0) -- (2.2,-2.2) -- (0,-2.2) -- cycle;
		\draw[line width=.25mm] (0,0) -- (2.2,0) -- (2.2,-2.2) -- (0,-2.2) -- cycle;
		
		\draw[line width=.25mm] (0,0) -- (2.2,0) -- (2.2,-2.2) -- (0,-2.2) -- cycle;
		\node[] at (2,-2) {};
		\node[draw=none, RWTH_red100] at (2,-2) {$\times$};
		
		\node[] at (1.5,-1.5) {};
		\node[] at (0.5,-0.5) {};
		
		\node[draw=none, RWTH_red100] at (1.5,-0.5) {$\times$};
		\node[draw=none, RWTH_red100] at (0.5,-1.5) {$\times$};
		

		\draw[RWTH_orange100, ->, thick] (2,-2) -- (2,-1.5);
		\draw[RWTH_orange100, ->, thick] (0.5,-1.5) -- (0.5,-2);
		
		\draw[RWTH_tuerkis100, ->, thick] (2,-1.5) -- (2,-0.5);
		\draw[RWTH_tuerkis100, ->, thick] (1.5,-0.5) -- (1.5,-1.5);

	\end{scope}
	\node[xshift=-.35cm, yshift=-1.8cm] at (boxC.west) {\textbf{(c)}};
\end{tikzpicture}
\caption{
Illustrations for the Monge and Anti-Monge case.
(a)
Modification to ensure that no cycles have length larger than four.
(b)
Modification to move $2$-cycles east of a $4$-cycle into a $4$-cycle.
(c)
Modification to move $2$-cycles southeast of a $4$-cycle into a $4$-cycle.
}
\label{fig:monge1}
\label{fig:monge2}
\label{fig:monge3}
\end{figure}

Based on this structural result we can easily compute an optimal solution for \probl{RecovAP} by solving the 
instance of the \probl{AP} on the subgraph stated in point~3 of \cref{thm:mongestructure} and 
then completing the perfect matchings $M_1$ and $M_2$ as stated in points~1 and~2. In summary,
we obtain the following result.
\begin{theorem}\label{thm:mongealg}
	Let $c_1$ be given by a Monge matrix $A$ and $c_2$ be given by an Anti-Monge matrix $B$.
	Then the \probl{RecovAP} can be solved in $O(n + k \log k)$ time.
\end{theorem}

In the following we prepare the proof of \cref{thm:mongestructure} based on several structural lemmas.
The main tool to analyze feasible solutions $M_1, M_2$ is their decomposition into 
$M_1$-$M_2$-alternating cycles of even length. For any even number $s$ we call such an 
alternating cycle an $s$-cycle. Using this language, we call an edge $e \in M_1 \cap M_2$ a $2$-cycle.
A $4$-cycle consists of four edges $\{ u_{i}, v_{j} \}, \{u_{i'}, v_{j'} \} \in M_1$ and 
$\{ u_{i}, v_{j'} \}, \{u_{j'}, v_{i} \} \in M_2$. Note that by the fact that $A$ is Monge 
and $B$ is Anti-Monge it is always optimal that $i < i'$ and $j < j'$. We call a $4$-cycle fulfilling this 
property an aligned $4$-cycle. We identify the 
$4$-cycle with its indices $(i,j,i',j')$. Similarly, we identify a $2$-cycle $\{v_i, v_j\} \in M_1 \cap M_2$ with its indices $(i,j)$.
Also, observe that 
in our matrix visualization $4$-cycles correspond to $2 \times 2$ submatrices where the corners 
diagonal to each other are marked by squares and crosses. The fact that $i<i'$ and $j<j'$ 
implies that the squares are drawn into the northwest and southeast corner and the crosses are drawn into 
the northeast and southwest corners.
We say that a $4$-cycle $(i_1, j_1, i'_1, j'_1)$ is nested inside another $4$-cycle $(i_2, j_2, i'_2, j'_2)$
if it holds that $i_2 < i_1 < i'_1 < i'_2$ and $j_2 < j_1 < j'_1 < j'_2$.
Analogously we say that a $2$-cycle $(i_1, j_1)$ is nested inside a $4$-cycle $(i_2, j_2, i'_2, j'_2)$
if $i_2 < i_1 < i'_2$ and $j_2 < j_1 < j'_2$.

Note, that in the language of such cycles \cref{thm:mongestructure} is equivalent to: 
there exists an optimal solution which consists of $\lfloor \frac{n-k}{2} \rfloor$ many aligned $4$-cycles and all the other matching edges are $2$-cycles; 
all $4$-cycles are nested into each other and the $2$-cycles are nested inside the innermost $4$-cycle; the $2$-cycles form a minimum 
cost perfect matching with respect to $c_1 + c_2$ on their vertices.

In \cref{lem:nolongcyc} we prove that there always exists an optimal solution without $s$-cycles for $s > 4$,
and all the $4$-cycles are nested and aligned.
The main idea here is to iteratively remove such long cycles from the outside to the inside.
In a second step (\cref{lemma:cycles:are:nested}) we then show that there exists an optimal solution in which 
all $2$-cycles lie inside the innermost $4$-cycle.

\begin{lemma}\label{lem:nolongcyc}
Let perfect matchings $M_1,M_2$ be feasible solutions to \probl{RecAP}.
Then there exists a solution $M_1',M_2'$ consisting only of 2-cycles and aligned 4-cycles, 
and all the $4$-cycles are nested.
\end{lemma}
\begin{proof}
	As a first step consider the subinstance (submatrices) where all vertices (rows and columns) 
	contained in $2$-cycles are removed.  Note, that this way the submatrices of $A$ and $B$ remain 
	Monge and Anti-Monge.
	
	Now assume that for $l = 1,\dots,\ell-1$ it holds that $(l,l,n+1-l,n+1-l)$ already forms 
	a $4$-cycle in $M_1, M_2$. We now construct matchings $M'_1, M'_2$ of smaller or equal cost 
	such that $(l,l,n+1-l,n+1-l)$ is also a $4$-cycle for $l=\ell$.
	If $\{u_{l}, v_{l}\} \notin M_1$ it holds that there are edges 
	$\{u_{i}, v_{l}\}, \{u_{l}, v_{j}\} \in M_1$. 
	By our assumption $l < i < n+1-l$.
	and since $A$ is Monge 
	we can exchange those edges for the edges $\{u_{l}, v_{l}\}, \{u_{i}, v_{j}\}$ in $M'_1$.
	See \cref{fig:monge1} for an illustration of this modification.
	Analogously, we can ensure that $M'_1$ also contains $\{u_{n+1-l}, v_{n+1-l}\}$ and 
	$M'_2$ contains both $\{u_{n+1-l}, v_{l}\}, \{u_{l}, v_{n+1-l}\}$, forming the $4$-cycle as claimed.
	We apply this process inductively. This way we obtain a solution consisting of only 
	nested aligned $4$-cycles, except for maybe one additional
	$2$-cycle exactly in the center of the matrix. We obtain the solution $M'_1, M'_2$ as 
	claimed by adding back the vertices (rows and columns) of the $2$-cycles removed in the first step.
\end{proof}

\begin{lemma}
	\label{lemma:cycles:are:nested}
	There is a solution \( M_1,M_2 \) with minimum cost \( c_1(M_1)+c_2(M_2) \) where all \( 4 \)-cycles are nested and aligned , and where no \( 2 \)-cycle is outside of a \( 4 \)-cycle.
\end{lemma}
\begin{proof}
	Note that the first part of the claim is already implied by \cref{lem:nolongcyc},
	and we start with matchings $M_1, M_2$ fulfilling the structure stated in \cref{lem:nolongcyc}.
	For the second claim we again process the cycles from outside to inside with respect to the 
	nesting order of the $4$-cycles. Assume that $(x,y)$ is the outmost $2$-cycle in $M_1,M_2$,
	and let $(i,j,i',j')$ be the outmost $4$-cycle that does not contain $(x,y)$.
	We show how to modify $M_1, M_2$ such that the cost does not increase,
	the number of $2$-cycles does not decrease and such that the number of $4$-cycles that 
	that contain all $2$-cycles is increased by one.

	We do this by looking at two distinct cases (up to symmetry).
	Case 1 is the case when $i < x < j$ and $i'<j'<y$, i.e. $(x,y)$ lies on the eastern side of $(i,j,i',j')$.
	In this case we remove $\{u_{i}, v_{j'}\}$ and $\{u_{x}, u_{y}\}$ from $M_2$ and add 
	$\{u_{x}, v_{j'}\}$ and $\{u_{i}, v_{y}\}$ to $M_2$, which can only improve the cost,
	since $B$ is Anti-Monge. In addition we remove $\{u_{i'}, v_{j'}\}$ and $\{u_{x}, v_{y}\}$ from $M_1$
	and add $\{u_{x}, v_{j'}\}$ and $\{u_{i'}, v_{y}\}$ to $M_1$. See \cref{fig:monge2} for this modification.
	Note, that now $(i,j,i',y)$ is a new $4$-cycle containing the new $2$-cycle $(x,j')$.
	A $2$-cycle lying to the north, south or west of the $4$-cycle can be handled similarly.

	Case 2 is the case when the $2$-cycle lies to the southeast of the $4$-cycle, i.e.\
	$i < i' < x$ and $j < j' < y$.
	In this case we remove $\{u_{i'}, v_{j}\}$ and $\{u_{x}, v_{y}\}$ from $M_2$ and replace 
	it with $\{u_{x}, v_{j}\}$ and $\{u_{i'}, v_{y}\}$ which can only decrease the cost by 
	the fact that $B$ is Anti-Monge. As a second step we remove $\{u_{i}, v_{j'}\}$ and $\{u_{i'}, v_{y}\}$
	from $M_2$ and add $\{u_{i'}, v_{i'}\}$ and $\{u_{i}, v_{y}\}$ to $M_2$. See \cref{fig:monge3} for this modification.
	Note, that now $(i,j,x,y)$ is a new $4$-cycle containing the new $2$-cycle $(i',j')$.
	The cases when the $2$-cycle lies northeast, southwest or northwest can be handled similarly.
\end{proof}

Now we are ready to give the proof of \cref{thm:mongestructure}.

\begin{proof}[Proof of \cref{thm:mongestructure}]
	Basically \cref{lemma:cycles:are:nested} already implies the combinatorial 
	structure claimed in \cref{thm:mongestructure}. The only point missing is that
	there are exactly $\lfloor \frac{n-k}{2} \rfloor$ many nested $4$-cycles in the 
	solution and the remaining edges form $2$-cycles inside.
	Note that selecting more $2$-cycles than strictly necessary (by the constraint or the 
	combinatorial structure) is never helpful, since the $4$-cycles correspond to the 
	optimal solution of the two independent \probl{AP}s.
	
	Hence, if 
	$2$ divides $n-k$ the claim follows directly, since this means exactly $k$ $2$-cycles 
	are chosen and the rest is chosen optimally with respect to the Monge and Anti-Monge matrix.

	In the other case, we can assume that $n$ is even and $k$ is odd, since if $n$ is odd an optimal solution
	always selects the edge $\{u_{(n+1)/2}, v_{(n+1)/2}\}$ as a $2$-cycle, giving an equivalent instance 
	with $n-1$ rows and columns and the constraint to select at least $k-1$ many $2$-cycles.
	Hence let $n$ be even and $k$ odd. By a simple parity argument we must select at least $k+1$ many $2$-cycles,
	and the claim follows.
\end{proof}

\section{The Second Stage Recoverable Assignment Problem}
\label{sec:incremental}

In this section we study a variant of \probl{RecAP} in which the perfect matching $M_1$ is fixed
and we are looking for a perfect matching $M_2$ of minimum linear cost $c_2(M_2)$ subject 
to the constraint that $|M_1 \cap M_2| \geq k$. 
Note that \probl{2S-RecovAP} is a special case of 
 \probl{RecovAP}, with the special cost structure $c_1(e) = 0$ if $e \in M_1$ and 
$c_1(e) = \infty$ otherwise.

Using the language of recoverable robust optimization this problem 
is called the incremental assignment problem.
Şeref et al.~\cite{seref-incremnetal-ap} study this problem and obtain a straightforward reduction 
to \probl{Exact Matching} in red-blue bipartite graphs,
one of the few natural problems known to be in RNC for which no polynomial time 
algorithm is known.

\problemdefSimple{\probl{Exact Matching in Red-Blue Bipartite Graphs}}
{A bipartite graph $G = (U \cup V, E)$, a subset of edges $R \subseteq E$ (the red colored edges), and an integer $k \in \nn$.}
{Is there a perfect matching $M$ of $G$ such that $|M \cap R| = k$.}

%
Mulmuley et al.~\cite{mulmuley1987matching} show that this problem can be solved 
in randomized polynomial time if all costs are polynomially bounded. 
In summary, the following result holds.
\begin{corollary}[Şeref et al.~\cite{seref-incremnetal-ap}]
	\probl{2S-RecovAP} can be solved by an $\textup{RNC}_2$ algorithm, if all costs $c_2$ 
	are polynomially bounded.
\end{corollary}
Note, that the techniques for this algorithm are not specific to bipartite graphs. 
Hence they can also be used to solve the second stage of the recoverable perfect matching problem on general graphs in $\text{RNC}_2$.

Surprisingly, we are able to prove that the complexity of these problems is essentially equal. 
We show that 
\probl{2S-RecovAP} is at least as hard as \probl{Exact Matching} in red-blue bipartite graphs.
Note here that our following logspace reduction implies a reduction in $\textup{NC}_2$~\cite{DBLP:books/daglib/0072413}.

\begin{theorem}
	\probl{Exact Matching in red-blue bipartite graphs} is logspace reducible to \probl{2S-RecovAP}.
\end{theorem}

\begin{proof}
	As a first part of the proof we give a reduction from exact matching in red-blue
	colored bipartite graphs to the special case of the problem where the set of red edges
	forms a matching. In the second part we then show that we can reduce this problem
	to the \probl{Second Stage of the RecovAP}.

	For this first part let $((G=(U \cup V), E), R, k)$ be the given instance of exact matching. We construct 
	a new bipartite graph $G'' = (U'' \cup V'', E'')$ and a set of edges $R'' \subseteq E''$ (see \cref{fig:exactmatching} for an illustration).
	Note that without loss of generality we can assume that $G$ contains no vertex of degree one, 
	since such vertices can always be preprocessed in a trivial way.

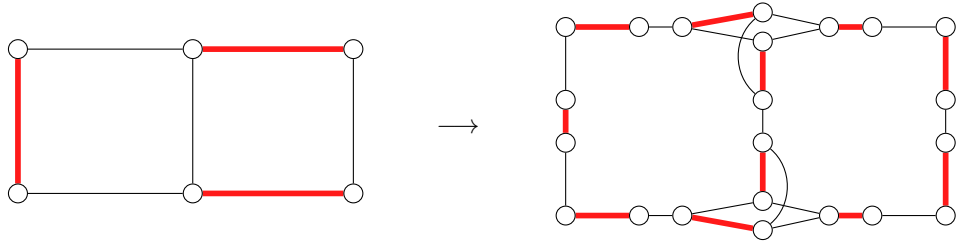
\begin{figure}
	\def\boxSize{2.5cm}
	\def\sideSep{0.2cm}
	\def\sideDist{0.25cm}
	\centering
	\begin{tikzpicture}[every node/.style={draw, circle, fill=white, minimum size=.25cm,inner sep=0pt}]
	\coordinate (tl);
	\coordinate[right=2*\boxSize of tl] (tr);
	\coordinate[below=\boxSize of tl] (bl);
	\coordinate[right=2*\boxSize of bl] (br);
	
	\coordinate[right=1.0*\boxSize of tr] (tl2);
	\coordinate[right=2*\boxSize of tl2] (tr2);
	\coordinate[below=\boxSize of tl2] (bl2);
	\coordinate[right=2*\boxSize of bl2] (br2);
	
	
	\node[draw=none, below right = 0.45*\boxSize and 0.35*\boxSize of tr] (thing) {$\longrightarrow$};
	
	\node[below right = \sideSep and \sideSep of tl] (gv1) {};
	\node[above right = \sideSep and \sideSep of bl] (gv2) {};
	\node[below right = \sideSep and \boxSize of tl] (gv3) {};
	\node[above right = \sideSep and \boxSize of bl] (gv4) {};
	\node[below left = \sideSep and \sideSep of tr] (gv5) {};
	\node[above left = \sideSep and \sideSep of br] (gv6) {};

	\node at (tl2) (rv1) {};
	\node at (bl2) (rv2) {};
	
	\node[above right = 0.5*\sideSep and \boxSize of tl2] (rv31) {};
	\node[below right = 0.5*\sideSep and \boxSize of tl2] (rv32) {};
	\node[above right = 0.5*\sideSep and \boxSize of bl2] (rv41) {};
	\node[below right = 0.5*\sideSep and \boxSize of bl2] (rv42) {};
	
	\node at (tr2) (rv5) {};
	\node at (br2) (rv6) {};
	
	\node[below = (1/3)*\boxSize of tl2] (ue1) {};
	\node[above = (1/3)*\boxSize of bl2] (ve1) {};
	
	\node[right = (1/3)*\boxSize of tl2] (ue2) {};
	\node[left = (4/3)*\boxSize of tr2] (ve2) {};
	
	\node[right = (1/3)*\boxSize of bl2] (ue3) {};
	\node[left = (4/3)*\boxSize of br2] (ve3) {};
	
	\node[below right = ((1.05/3)*\boxSize) and \boxSize of tl2] (ue4) {};
	\node[above right = ((1.05/3)*\boxSize) and \boxSize of bl2] (ve4) {};
	
	\node[left = (1/3)*\boxSize of tr2] (ue5) {};
	\node[right = (4/3)*\boxSize of tl2] (ve5) {};
	
	\node[left= (1/3)*\boxSize of br2] (ue6) {};
	\node[right = (4/3)*\boxSize of bl2] (ve6) {};
	
	\node[below = (1/3)*\boxSize of tr2] (ue7) {};
	\node[above = (1/3)*\boxSize of br2] (ve7) {};
	
	\begin{pgfonlayer}{bg}
	\draw[matched] (gv1) -- (gv2);
	\draw (gv1) -- (gv3);
	\draw (gv2) -- (gv4);
	\draw (gv3) -- (gv4);
	\draw[matched] (gv3) -- (gv5);
	\draw[matched] (gv4) -- (gv6);
	\draw (gv5) -- (gv6);
	
	\draw[matched] (ue1) -- (ve1);
	\draw[matched] (ue5) -- (ve5);
	\draw[matched] (ue6) -- (ve6);
	
	\draw[matched] (rv1) -- (ue2);
	\draw[matched] (rv31) -- (ve2);
	
	\draw[matched] (rv2) -- (ue3);
	\draw[matched] (rv42) -- (ve3);
	\draw (rv42) to[bend right=50] (ve4);
	
	\draw (rv31) to[bend right=50] (ue4);
	\draw[matched] (rv32) -- (ue4);
	\draw[matched] (rv41) -- (ve4);
	
	\draw[matched] (rv5) -- (ue7);
	\draw[matched] (rv6) -- (ve7);
	
	\draw (rv1) -- (ue1);
	\draw (rv2) -- (ve1);
	\draw (ue2) -- (ve2);
	\draw (ue3) -- (ve3);
	\draw (ue4) -- (ve4);
	\draw (rv32) -- (ve2);
	\draw (rv31) -- (ve5);
	\draw (rv32) -- (ve5);
	\draw (rv41) -- (ve3);
	\draw (rv41) -- (ve6);
	\draw (rv42) -- (ve6);
	\draw (rv5) -- (ue5);
	\draw (rv6) -- (ue6);
	\draw (ue7) -- (ve7);
	\end{pgfonlayer}
	\end{tikzpicture}
	
	\caption{Visualization of the reduction from the variant where $R$ is a matching to general exact matching.}
	\label{fig:exactmatching}
\end{figure}

	For every vertex 
	$v$ in $G$  we introduce an independent set $v_1, \dots, v_{\deg(v)-1}$ of $\deg(v)-1$ many vertices.
	For every edge $e = \{u,v\}$ of $G$ we introduce two vertices $u_e$ and $v_e$ in $G''$
	and connect them by the edge $\{u_e, v_e\}$. If $\{u,v\} \in R$ then we add $\{u_e, v_e\}$ to $R''$.
	In addition we add all the edges $\{u_{i}, u_e\}$ for $i=1,\dots,\deg(u)-1$ and $\{v_{i}, v_e\}$ for $i=1,\dots,\deg(v)-1$ to $E''$.
	Observe that the graph $G''$ is bipartite if and only if $G$ is bipartite and note that since $|U| = |V|$ also $|U''| = |V''|$.

	\begin{claim}
		There exists a perfect matching of $G$ with exactly $k$ edges in $R$ if and only 
		if there exists a perfect matching in $G''$ with exactly $k$ edges in $R''$.
	\end{claim}
	\begin{claimproof}
		Given a perfect matching $M$ of $G$ with exactly $k$ edges in $R$ we construct a perfect matching $M''$
		in $G''$ with exactly $k$ edges in $R''$. For each $e \in M$ we add the edge $\{v_e, u_e\}$ to $M''$.
		Note that this way we add exactly $k$ edges from $R''$ to $M''$.
		Now for every vertex $v$ exactly one of its incident edges in $G$ is in $M$.
		Hence there are exactly $\deg(v)-1$ incident edges that are not in $M$.
		For each such edge $\{v,w\} \in E \setminus M$ we select one of the vertices 
		$v_j$ for $j \in \{1, \dots, \deg(v)-1\}$ and add $\{v_j, v_e\}$ to $M''$.
		Note that this way $M''$ is a perfect matching in $G''$ with exactly $k$ edges from 
		$R''$ in $M''$.

		For the converse direction, assume that $M''$ is a perfect matching of $G''$ 
		with exactly $k$ edges from $R''$. We construct a perfect matching $M$ of $G$ 
		with exactly $k$ edges from $R$.
		Note that every original vertex $v \in V$ is replaced by an independent $v_1, \dots, v_{\deg(v)-1}$ in $G''$.
		Each of these vertices is matched 
		to an vertex $v_{e'}$ for an incident edge $e' \in E$. But since there 
		exist only $\deg(v)-1$ such vertices and $v$ has exactly $\deg(v)$ incident 
		edges in $G$ there exists a unique edge $e = \{v,w\}$ in $G$ for which 
		$v_{e}$ is not matched to one of those vertices. Hence $v_{e}$ must be 
		matched to its only remaining neighbor $w_{e}$ by $M''$. Note, that 
		by similar arguments $e$ is also the unique incident edge $e'$ to $w$ 
		for which the vertex $w_{e'}$ is not matched to one of the vertices 
		$w_{1}, \dots, w_{\deg(w)-1}$. We add the edge $\{v,w\}$ to $M$. From the 
		arguments above it is clear that $M$ is a perfect matching in $G$. Since 
		$M''$ contains exactly $k$ edges from $R''$ and for each edge 
		$\{v_e, w_e\}$ in $M''$ we add the edge $e=\{v,w\}$ to $M$, also $M$ 
		contains exactly $k$ edges from $R$.
	\end{claimproof}

\medskip

	As a next step we show how to obtain the instance $(G', M_1, c_2, k)$ of the \probl{Second Stage of the RecovAP}
	such that there exists a perfect matching of $G$ with exactly $k$ edges in $R$ if and only if 
	the optimal value for $(G' = (U' \cup V', E'), M_2, c_2, k)$ is $k$. The graph 
	$G''$ constructed above is a subgraph of $G'$. In addition to that, for each vertex 
	$v \in V''$ that is not matched by $R'$ we add an additional vertex $v'$ to $G'$ and
	the edge $\{v,v'\}$.
Since $G''$ is a bipartite graph with $|U''| = |V''|$ and $R''$ is a matching 
	we can select for each such $v'$ another unique vertex $u'$ and add the edge $\{v',u'\}$ to $G'$.
	We define the set $R'$ as the set of edges consisting of $R''$ and all edges 
	$\{v,v'\}$ for all $v \in V'' \setminus R''$. Note that $R'$ is a perfect matching 
	in $G'$. We set $c_2(e) = 1$ for all $e \in R''$ and $c_2(e) = \infty$ for all edges
	$\{v,v'\}$ where $v \in V'' \setminus R''$. For all other edges the cost $c_2$ is equal to $0$.

	\begin{claim}
		There exists a perfect matching of $G$ with exactly $k$ edges in $R$ if and only 
		if there exists solution to the \probl{Second Stage of the RecovAP} instance $(G', M_1, c_2, k)$ with cost $k$.
	\end{claim}
	\begin{claimproof}
		Assume that there exists a perfect matching $M$ of $G$ with exactly $k$ edges in $R$.
		Then by the claim above there also exists a perfect matching $M''$ of $G''$ with exactly 
		$k$ edges in $G''$. Based on $M''$ we define a perfect matching $M_2$ in $G'$ in the following way.
		The matching $M''$ is added to $M_2$ and hence all vertices in the subgraph $G''$ of $G'$ 
		are matched. For all the remaining vertices, by the construction above there exists 
		a unique matching consisting of the edges $\{v', u'\}$ which are added to $M_2$. Note, 
		that $c_2(M_2) = k$.

		For the converse direction, assume that $M_2$ is a perfect matching in $G''$ with at 
		least $k$ edges from $M_1$ and cost $k$. Since only edges in $R' \cap M_1$ 
		have finite cost and $c_2(M_2) = k$  it holds that exactly $k$ edges from $R'$ are 
		contained in $M_2$. In addition, since none of the edges $\{v, v'\}$ are contained in 
		$M_2$ it holds that $M_2 \cap E'$ is a perfect matching in $G'$ with exactly $k$ edges from $R'$. Hence by the first claim 
		there exists a perfect matching in $G$ containing exactly $k$ edges from $R$.
	\end{claimproof}
	
	This completes the reduction as claimed in the theorem.
	Note that this reduction can be implemented using logarithmic space. 
	We just have to process one vertex after another and need to implement a counter counting up to the degree of a vertex.
\end{proof}

The case with costs $c_2$ that are not polynomially bounded remains open. But note,
that a RNC algorithm for this problem would imply an RNC algorithm for 
the special case of obtaining a minimum cost perfect matching,
which is a long standing open problem~\cite{exactmatchingweb}.

\bibliography{robust}

\begin{thebibliography}{10}

\bibitem{exactmatchingweb}
Exact matching in red-blue bipartite graphs.
\newblock
  \url{http://lemon.cs.elte.hu/egres/open/Exact_matching_in_red-blue_bipartite_graphs}.
\newblock Accessed: 2020-07-13.

\bibitem{DBLP:journals/actaC/Bodlaender93}
Hans~L. Bodlaender.
\newblock A tourist guide through treewidth.
\newblock {\em Acta Cybern.}, 11(1-2):1--21, 1993.
\newblock URL:
  \url{http://www.inf.u-szeged.hu/actacybernetica/edb/vol11n1\_2/Bodlaender\_1993\_ActaCybernetica.xml}.

\bibitem{DBLP:journals/siamcomp/Bodlaender96}
Hans~L. Bodlaender.
\newblock A linear-time algorithm for finding tree-decompositions of small
  treewidth.
\newblock {\em {SIAM} J. Comput.}, 25(6):1305--1317, 1996.
\newblock \href {https://doi.org/10.1137/S0097539793251219}
  {\path{doi:10.1137/S0097539793251219}}.

\bibitem{burkard2012assignment}
Rainer Burkard, Mauro Dell'Amico, and Silvano Martello.
\newblock {\em Assignment Problems: Revised Reprint}.
\newblock SIAM, 2012.

\bibitem{burkard1996perspectives}
Rainer~E. Burkard, Bettina Klinz, and R{\"{u}}diger Rudolf.
\newblock Perspectives of {M}onge properties in optimization.
\newblock {\em Discret. Appl. Math.}, 70(2):95--161, 1996.
\newblock \href {https://doi.org/10.1016/0166-218X(95)00103-X}
  {\path{doi:10.1016/0166-218X(95)00103-X}}.

\bibitem{busing2012recoverable}
Christina B{\"u}sing.
\newblock Recoverable robust shortest path problems.
\newblock {\em Networks}, 59(1):181--189, 2012.

\bibitem{Cygan2015}
Marek Cygan, Fedor~V. Fomin, Łukasz Kowalik, Daniel Lokshtanov, Daniel Marx,
  Marcin Pilipczuk, Michał Pilipczuk, and Saket Saurabh.
\newblock Parameterized algorithms, 2015.
\newblock URL: \url{http://ebooks.ciando.com/book/index.cfm/bok_id/1960687}.

\bibitem{Frieze1983}
Alan~M Frieze.
\newblock Complexity of a 3-dimensional assignment problem.
\newblock {\em European Journal of Operational Research}, 13(2):161--164, 1983.

\bibitem{gurjar2012planarizing}
Rohit Gurjar, Arpita Korwar, Jochen Messner, Simon Straub, and Thomas Thierauf.
\newblock Planarizing gadgets for perfect matching do not exist.
\newblock In {\em International Symposium on Mathematical Foundations of
  Computer Science}, pages 478--490. Springer, 2012.

\bibitem{hradovich2017recoverable-MST}
Mikita Hradovich, Adam Kasperski, and Pawe{\l} Zieli{\'n}ski.
\newblock Recoverable robust spanning tree problem under interval uncertainty
  representations.
\newblock {\em Journal of Combinatorial Optimization}, 34(2):554--573, 2017.

\bibitem{hradovich2017recoverable}
Mikita Hradovich, Adam Kasperski, and Pawe{\l} Zieli{\'n}ski.
\newblock The recoverable robust spanning tree problem with interval costs is
  polynomially solvable.
\newblock {\em Optimization Letters}, 11(1):17--30, 2017.

\bibitem{iwamasa2020conf}
Yuni Iwamas and Kenjiro Takayawa.
\newblock Optimal matroid bases with intersection constraints: Valuated
  matroids, m-convex functions, and their applications.
\newblock In {\em Proceedings of the 16th Annual Conference on Theory and
  Applications of Models of Computation (TAMC 2020), to appear}, 2020.

\bibitem{kasperski2015robust}
Adam Kasperski and Pawe{\l} Zieli{\'n}ski.
\newblock Robust recoverable and two-stage selection problems.
\newblock {\em Discrete Applied Mathematics}, 233:52--64, 2017.

\bibitem{lachmann2019}
Thomas Lachmann and Stefan Lendl.
\newblock Efficient algorithms for the recoverable (robust) selection problem.
\newblock In {\em Proceedings of the 17th Cologne-Twente Workshop on Graphs and
  Combinatorial Optimization}, 2019.

\bibitem{lendl2019matroid}
Stefan Lendl, Britta Peis, and Veerle Timmermans.
\newblock Matroid bases with cardinality constraints on the intersection.
\newblock {\em arXiv preprint arXiv:1907.04741}, 2019.

\bibitem{mulmuley1987matching}
Ketan Mulmuley, Umesh~V. Vazirani, and Vijay~V. Vazirani.
\newblock Matching is as easy as matrix inversion.
\newblock {\em Combinatorica}, 7(1):105--113, 1987.
\newblock \href {https://doi.org/10.1007/BF02579206}
  {\path{doi:10.1007/BF02579206}}.

\bibitem{DBLP:books/daglib/0072413}
Christos~H. Papadimitriou.
\newblock {\em Computational complexity}.
\newblock Addison-Wesley, 1994.

\bibitem{DBLP:journals/jal/RobertsonS86}
Neil Robertson and Paul~D. Seymour.
\newblock Graph minors. {II.} algorithmic aspects of tree-width.
\newblock {\em J. Algorithms}, 7(3):309--322, 1986.
\newblock \href {https://doi.org/10.1016/0196-6774(86)90023-4}
  {\path{doi:10.1016/0196-6774(86)90023-4}}.

\bibitem{seref-incremnetal-ap}
Onur Şeref, Ravindra~K. Ahuja, and James~B. Orlin.
\newblock Incremental network optimization: Theory and algorithms.
\newblock {\em Operations Research}, 57(3):586--594, 2009.
\newblock \href {https://doi.org/10.1287/opre.1080.0607}
  {\path{doi:10.1287/opre.1080.0607}}.

\end{thebibliography}

\newpage

\appendix

\section{An Application in Recoverable Robust Optimization}\label{sec:robustapplication}

The \probl{RecAP} studied in this paper has direct applications to the recoverable robust optimization problem 
studied by the robust optimization community.
In recoverable robust optimization we are given a general discrete optimization 
problem $(E, \mathcal{F})$, where $E$ is the ground set and $\mathcal{F}$ is the family of 
subsets of $E$ giving the feasible solutions. In addition we are given the recoverability parameter $k' \in \nn$,
a linear first stage cost function $c_1$ and an uncertainty set $\mathcal{U}$ that contains different scenarios $S$, where each scenario $c_S \in \mathcal{U}$ gives a possible linear second stage cost function $c_{S}$.

The recoverable robust optimization problem then consists of two stages:
In the first stage, one needs to select a feasible solution $X \in \mathcal{F}$.
Then, after the actual scenario $c_S \in \mathcal{U}$ is revealed, there is a second stage, where 
a second feasible solution $Y \in \mathcal{F}$ is picked with the goal to minimize the worst-case cost $c_1(X) + c_{S}(Y)$ under the constraint that $Y$ 
differs in at most $k'$ elements from the original basis $Y$. This second stage minimization problem is called the incremental problem in the robust optimization literature, since only an incremental change to $X$ is allowed.
This means, we require that $Y$ satisfies $|Y \setminus X| \leq k'$.
The recoverable robust optimization problem can be written as follows:

\begin{equation} \label{eq:recrobopt}
\min_{X \in \mathcal{F}} \left( c_1(X) \;+\; \max_{c_S \in \mathcal{U}} \min_{\substack{Y \in \mathcal{F} \\ |Y \setminus X| \leq  k'} } c_S(Y) \right). 
\end{equation}

There are several ways in which the uncertainty set $\mathcal{U}$  can be represented. One popular way is the \emph{interval uncertainty representation}. In this representation, we are given functions $c': E \rightarrow \rr$, $d: E \rightarrow \rr_+$ and assume that the uncertainty set $\mathcal{U}$ can be represented by a set of $|E|$ intervals:
\[
\mathcal{U} = \left\{ c_S \mid c_S(e) \in [c'(e), c'(e) + d(e)], \; e \in E \right\}.
\]

In the worst-case scenario $\bar{S}$ we have for all $e \in E$ that $c_{\bar{S}}(e) =c'(e) + d(e)$.
Kasperski and Zieli\'{n}ski~\cite{kasperski2015robust} observed that if all feasible solutions in $\mathcal{F}$ 
have the same cardinality $r$, then the constraint $|Y \setminus X| \leq k'$ is equivalent to $|X \cap Y| \geq r-k'$. 
If $\mathcal{F}$ is the set of perfect matchings of a complete bipartite graph we get the recoverable robust assignment problem with interval uncertainty representation, for which $r$ is equal to $n$.
Observe that by Setting $c_2 = c_{\bar{S}}$ an optimal solution to \probl{RecovAP} is an optimal solution to the recoverable robust assignment problem with interval uncertainty representation.
%
Note that the \probl{2S-RecovAP}, studied in \cref{sec:incremental}, 
is exactly the incremental assignment problem studied in robust optimization.

\section{A Dynamic Program for Graphs of Bounded Treewidth}\label{sec:tw}
\newcommand{\MM}{\mathcal{M}}

We show that recoverable matching is FPT when parameterized by the treewidth of the input graph.
Note that the intersection size $k$ is not a parameter here.
We give a dynamic program that based on a nice tree decomposition computes minimum cost matchings with intersection at least $k$.

\begin{definition}[\cite{DBLP:journals/jal/RobertsonS86}]
A tree decomposition of a graph $G$ is a pair $(T,\beta)$ where $T$ is a tree (where we refer to $i \in V(T)$ as nodes) and $\beta$ mapping from the nodes of $T$ to sets of vertices of $G$ such that
\begin{itemize}
\item
$ \bigcup_{i \in V(T)} G[\beta(i)]=G$, and
\item
for every vertex $v\in V(G)$, the set of nodes $\{ i \in V(T) \mid i \in \beta(v) \}$ is connected.
\end{itemize}
The width of $(T,\beta)$ is the size $\max_{i \in V(T)} |\beta(i)|-1$.
The treewidth is the minimum width of all tree decompositions of $G$.
\end{definition}

The treewidth can be computed in FPT \cite{DBLP:journals/siamcomp/Bodlaender96},
	and we assume in the following that we are given a tree decomposition of minimal width.
We may further work with a nice tree decomposition: 
Here leave nodes $i$ have $\beta(i) = \emptyset$, and $T$ has a root $r$ that is a leaf.
Further branching nodes have exactly two children $j_1$ and $j_2$ where $\beta(i)=\beta(j_1)=\beta(j_2)$.
Finally, a node $i$ with exactly one child $j$ is either an introduce node, where $\beta(i)= \beta(j) \cup \{v\}$ for some $v\in V(G)$, or a forget node, where $\beta(i)= \beta(j) \setminus \{v\}$ for some $v\in V(G)$.
One can easily extend a tree decomposition to a nice tree decomposition in time $O( \mbox{tw}\cdot n )$ (see for example \cite{DBLP:journals/actaC/Bodlaender93}).

Our dynamic program determines for each bag $\beta(i)$ given some constraints on $G[\beta(i)]$ an minimum cost solution for the current subgraph $G_i$.
Here $G_i$ is the subgraph induced by $\bigcup_{j \in \beta(T_i)} \beta(j)$ where $T_i$ is the subtree of $T$ rooted at $i$.
We use a \emph{preliminary matching} $M_W$ for the vertices $W$ of a bag $\beta(i)$ that may contain edges and singleton sets of vertices.
Formally $M_W \in\MM_{W}$ where $\MM_W$ is the family of disjoint sets of vertices $W$ of size one or two.
\begin{itemize}
\item
$\{v,w\}\in M_{\beta(i)}$ indicates a matching edge,
and we search for a matching in the subgraph that in particular contains $\{v,w\}$.
\item
$\{v\}\in M_{\beta(i)}$ indicates a vertex $v$ that is not yet matched,
and we search for a matching in the subgraph that in particular matches $\{v\}$ to some vertex $w\notin {\beta(i)}$.
\item
$v \cap \{ S \in M_{\beta(i)} \}=\emptyset$ indicates that $v$ is matched by some vertex not in ${\beta(i)}$,
and we search for a matching in the subgraph that in particular does not match $v$.
\end{itemize}
Consequently, the containment of a vertex $v$ in a bags $\beta(i)$ for one matching on a path from the root to a leaf appears as follows:
A vertex appears in a `forget node' (as we proceed downwards).
There it is added as $\{v\}$ or is already matched to some $w$ resulting in set $\{v,w\}$.
Sooner or later indeed $v$ should be matched to some $w$.
Then eventually either $v$ or $w$ disappears in a `introduce node'.
If $v$ disappears first, we should be able to remove an edge $\{u,w\}$.
If $w$ disappeared first, we should not find any set $\{v\}$.
Thus at an `introduce node' of $v$ we forbid the singleton $\{v\}$.

Note that we handle not one but two containments in matchings, namely $M_1$ and $M_2$.
Thus, as a table entry for our dynamic program, let
$$C_i[M_1,\;M_2,\;k'],~~~~\text{for }~~ M_1, M_2 \in \MM_{\beta(i)},~~ k'\in\{0,\dots,k\}$$
be the minimum cost $c(M_1^\star)+c(M_2^\star)$ of matchings $M_1^\star,M_2^\star$ in the subgraph $G_i$ extending the given preliminary matchings $M_1,M_2$ with intersection $|M_1^\star\cap M_2^\star| \geq k'$.
Here a matching $M_1^\star$ extends the preliminary matching $M_1$ if for every edge $\{v,w\}\in M_1$ appears in $M_1^\star$ and $M_1^\star$ matches every vertex in $V(G_i)$ but $v \in \beta(i)$ where $v \cap \{ S \mid S \in M_W \}=\emptyset$.

We compute the values of $C_i$ bottom up for the nodes $i$ of our tree decomposition.
Eventually we obtain the minimum cost of the recoverable assignment problem by evaluating $C_r(\emptyset,\emptyset,k)$ for the root $r$, which is a leaf node.

For a leaf node $i\neq r$ we only have to consider $M_1 = \emptyset = M_2$.
Then $C_i[\emptyset,\emptyset,0]=0$ and in any other case $C_i[\emptyset,\emptyset,k']=\infty$ for $k' > 0$.
Otherwise for node $i$, consider every subset $M_1, M_2 \subseteq \beta(i)$ and every $k'\in\{0,\dots,k\}$ and compute the minimum cost depending on whether the node is a introduce node, a forget node or a join node:
\begin{itemize}

\item
Case, $i$ is an forget node, and has one child node $j$.
That is there is one forgot vertex $v \in \beta(j)\setminus\beta(i)$.
As discussed before, $\{v\}\in M_1$, or $v$ is matched already to a singleton $\{w\}\in M_1$, likewise for the other matching $M_2$.
In case $M_1$ and $M_2$ match $v$ to a common $w$, we yield a matching edge in the intersection.

\begin{align*}
 C_i[M_1, M_2, k'] = 
 \min \!
\begin{dcases}
\min_{\{w\} \in M_1 \cap M_2}&
	\big\{ C_j[ M_1 \cup \{v,w\} \setminus \{w\},\; M_2 \cup \{v,w\} \setminus \{w\}, \; k'-1 ] \big\}, \\ \!
\min_{\substack{\{w_1\} \in M_1, \\ \{w_2\} \in M_2}}&
	\big\{ C_j[ M_1 \cup \{v,w_1\} \setminus \{w_1\},\; M_2 \cup \{v,w_2\} \setminus \{w_2\}, \; k' ] \big\},\\ \!
\min_{\{w_1\} \in M_1}&
	\big\{ C_j[ M_1 \cup \{v,w_1\} \setminus \{w_1\},\; M_2 \cup \{v\}, \; k' ] \big\},\\ \!
\min_{\{w_2\} \in M_1}&
	\big\{ C_j[ M_1 \cup \{v\},\; M_2 + \{v,w_2\} \setminus \{w_2\}, \; k' ] \big\},\\ \!
&
	\big\{ C_j[ M_1 \cup \{v\},\; M_2 \cup \{v\},\; k' ] \big\}.
\end{dcases}
\end{align*}

\item
Case, $i$ is an introduce node, and has one child node $j$.
That is there is one introduced vertex $v \in \beta(i)\setminus\beta(j)$.
Here we assure that $v$, which is no longer considered for $C_j$, is indeed matched.
that means $\{v\}\notin M_1$ and $\{v\}\notin M_2$ as otherwise $v$ found no matching partner.
\begin{equation*}
 C_i[M_1, M_2, k'] =  
\begin{cases}
C_j[ M_1' \cap \MM_{\beta(j)} ,\;  M_2' \cap \MM_{\beta(j)},\; k'],& \text{if } \{v\} \notin M_1, \{v\} \notin M_2, \\
\infty,& \text{else.}
\end{cases}
\end{equation*}

\item
Case, $i$ is a join node, and has two children $j_1$ and $j_2$.
Here we have to find solutions whose intersections $k^{(1)}$ and $k^{(2)}$ sum up to $k'+|I|$ where $I=\{ \{v,w\} \mid \{v,w\} \in M_1 \cap M_2 \}$ since the intersecting edges of the current bag are counted twice.
Further for a singleton vertex $\{w\}\in M_1$ only one subtree should contain a matching $w$.
Hence let $M^{(1)} \oplus M^{(2)} = M$ denote that $M^{(1)} \cup M^{(2)} = M$ and for $e\in M$ we have: $e \in M^{(1)} \cup M^{(2)}$ if and only if $|e|=2$.
\begin{equation*}
 C_i[M_1, M_2, k'] = 
\min_{\substack{k_1 + k_2 \; = \; k' + |I|, \\ M_1^{(1)} \oplus M_1^{(2)} = M_1, \\M_2^{(1)} \oplus M_2^{(2)} = M_2}} \;
C_{j_1}[ M_1^{(1)} ,\;  M_2^{(1)},\; k^{(1)}]
\; + \;
C_{j_2}[ M_1^{(2)} ,\;  M_2^{(2)},\; k^{(2)}]  \; - \; |I|.
\end{equation*}

\end{itemize}

It remains to consider the runtime.
Let $w$ be the treewidth plus one, hence the maximum bag size.
We maintain a table for values $C_i$ for every node $i$, which cover every choice of $M_1,M_2,k'$.
How many possible preliminary matchings $M \in \MM_{\beta(i)}$ are there?
We can bound the number by guessing the vertices that occur in a set of $M$, that are possibilities $O(2^w)$, and among those guessing a mapping to a potential matching partner, that are $O(2^{w \log w})$ possibility; hence $O(2^{w \log w})$ in total.

Now we consider the time spend at each type of node.
For the leaf and introduce nodes no computation is needed.
For a forget node we have to at most try all combinations of nodes $\{w_1\},\{w_2\}$ that possibly are in respective matchings $M_1, M_2$, which are $O(w^2)$ combinations to try.
Finally for a join node we have to cycle every partition of $k'$ and twice the possible separations of a preliminary matching $M\in \MM_{\beta(i)}$ into preliminary matchings $M_1$ and $M_2$.
There only for the singleton sets $\{w\}\in M$ we have to decide whether $\{w\}\in M_1$ or $\{w\}\in M_2$.
In total this yields $O(n 2^{2w})$ combinations to try for a join node.
Because the table size is FPT and the computation for each table entry as discussed is FPT,
	it follows that recoverable matching with parameter treewidth is FPT.

\end{document}